\documentclass[usenatbib,fleqn]{mnras}

\usepackage{newtxtext,newtxmath}
\usepackage{dsfont}

\usepackage[T1]{fontenc}
\usepackage{ae,aecompl}

\usepackage{natbib}

\usepackage{graphicx}	
\usepackage{amsmath}	
\usepackage{amssymb}	
\usepackage{mathtools}
\usepackage{breqn}
\usepackage[flushleft]{threeparttable}
\usepackage{enumitem}
\usepackage{algorithm}
\usepackage{algorithmicx}
\usepackage{algpseudocode}
\usepackage{bm}
\usepackage{tabularx}
\usepackage{adjustbox}

\usepackage{scalerel}
\usepackage{tikz}
\usetikzlibrary{svg.path}

\definecolor{orcidlogocol}{HTML}{A6CE39}
\tikzset{
  orcidlogo/.pic={
    \fill[orcidlogocol] svg{M256,128c0,70.7-57.3,128-128,128C57.3,256,0,198.7,0,128C0,57.3,57.3,0,128,0C198.7,0,256,57.3,256,128z};
    \fill[white] svg{M86.3,186.2H70.9V79.1h15.4v48.4V186.2z}
                 svg{M108.9,79.1h41.6c39.6,0,57,28.3,57,53.6c0,27.5-21.5,53.6-56.8,53.6h-41.8V79.1z M124.3,172.4h24.5c34.9,0,42.9-26.5,42.9-39.7c0-21.5-13.7-39.7-43.7-39.7h-23.7V172.4z}
                 svg{M88.7,56.8c0,5.5-4.5,10.1-10.1,10.1c-5.6,0-10.1-4.6-10.1-10.1c0-5.6,4.5-10.1,10.1-10.1C84.2,46.7,88.7,51.3,88.7,56.8z};
  }
}

\newcommand\orcidicon[1]{\href{https://orcid.org/#1}{\mbox{$^{\scalerel*{
\begin{tikzpicture}[yscale=-1,transform shape]
\pic{orcidlogo};
\end{tikzpicture}
}{+}}$}}}

\newcommand{\re}{{\rm Re}}
\newcommand{\im}{{\rm Im}}
\newcommand{\dd}{\mathrm{d}}

\newcommand{\sgn}{\mathrm{sgn}}
\newcommand{\erfi}{\mathrm{erfi}}
\newcommand{\erf}{\mathrm{erf}}
\newcommand{\nonumberoff}{}

\numberwithin{equation}{section}

\usepackage{etoolbox}
\makeatletter

\patchcmd{\NAT@citex}
  {\@citea\NAT@hyper@{%
     \NAT@nmfmt{\NAT@nm}%
     \hyper@natlinkbreak{\NAT@aysep\NAT@spacechar}{\@citeb\@extra@b@citeb}%
     \NAT@date}}
  {\@citea\NAT@nmfmt{\NAT@nm}%
   \NAT@aysep\NAT@spacechar\NAT@hyper@{\NAT@date}}{}{}

\patchcmd{\NAT@citex}
  {\@citea\NAT@hyper@{%
     \NAT@nmfmt{\NAT@nm}%
     \hyper@natlinkbreak{\NAT@spacechar\NAT@@open\if*#1*\else#1\NAT@spacechar\fi}%
       {\@citeb\@extra@b@citeb}%
     \NAT@date}}
  {\@citea\NAT@nmfmt{\NAT@nm}%
   \NAT@spacechar\NAT@@open\if*#1*\else#1\NAT@spacechar\fi\NAT@hyper@{\NAT@date}}
  {}{}

\makeatother

\DeclarePairedDelimiter\abs{\lvert}{\rvert}%
\DeclarePairedDelimiter\norm{\lVert}{\rVert}%

\makeatletter
\let\oldabs\abs
\def\abs{\@ifstar{\oldabs}{\oldabs*}}
\let\oldnorm\norm
\def\norm{\@ifstar{\oldnorm}{\oldnorm*}}
\makeatother

\title[Unified lensing and kinematic analysis]{Unified lensing and kinematic analysis for \textit{any} elliptical mass profile}

\author[A. J. Shajib]{
Anowar J. Shajib \orcidicon{0000-0002-5558-888X}$^{1}$\thanks{E-mail: ajshajib@astro.ucla.edu}
\\ \\
$^{1}$Department of Physics and Astronomy, University of California, Los Angeles, CA 90095-1547, USA \\}

\date{Accepted XXX. Received YYY; in original form ZZZ}

\pubyear{2019}

\begin{document}

\label{firstpage}
\pagerange{\pageref{firstpage}--\pageref{lastpage}}
\maketitle

\begin{abstract}
We demonstrate an efficient method to compute the strong-gravitational-lensing deflection angle and magnification for any elliptical surface-density profile. This method solves a numerical hurdle in lens modelling that has lacked a general solution for nearly three decades. The hurdle emerges because it is prohibitive to derive analytic expressions of the lensing quantities for most elliptical mass profiles. In our method, we first decompose an elliptical mass profile into Gaussian components. We introduce an integral transform that provides us with a fast and accurate algorithm for the Gaussian decomposition. We derive analytic expressions of the lensing quantities for a Gaussian component. As a result, we can compute these quantities for the total mass profile by adding up the contributions from the individual components. This lensing analysis self-consistently completes the kinematic description in terms of Gaussian components presented by \citet{Cappellari08}. Our method is general without extra computational burden unlike other methods currently in use. 
\end{abstract}




\section{Introduction: precise computation made easy}

Gravitational lensing (hereafter, lensing) has versatile applications in astrophysics and cosmology. Lensing is the effect when light bends while passing by a massive object. If two galaxies sit along the same line-of-sight of an observer, the background galaxy appears multiple times due to lensing. This system is called a galaxy-scale strong-lensing system (hereafter, lens). Lenses are useful to measure the Hubble constant $H_0$, dark matter subhalo mass function, dust characteristics in galaxies, mass of super-massive black holes, stellar initial mass function, etc. \citep[e.g.,][]{Falco99, Peng06, Suyu13, Vegetti14, Schechter14}.

%

In several applications, stellar kinematics play a complementary role to lensing. Lensing-only observables suffer from mass-sheet degeneracy (MSD) -- if we appropriately rescale a mass profile after adding an infinite mass-sheet on top of it, then all the lensing observables stay invariant except the time delay \citep{Falco85, Schneider14}. When a measured quantity depends on the mass profile, the MSD adds uncertainty to the measurement. Kinematics help break this degeneracy. Lensing observables probe the projected mass; kinematic observables probe the three-dimensional potential. Thus, the lensing--kinematics combination tightly constrains the mass profile, and enables us to robustly measure astrophysical and cosmological quantities \citep[e.g.,][]{Treu04, Barnabe11}.

In practice, it is difficult to compute lensing and kinematic quantities for elliptical mass profiles, which are required to describe the most common kind of lenses, i.e., elliptical galaxies.
We need to integrate over density profiles while fitting a model to either the lensing or the kinematic data. For example, in lensing, the deflection angle is an integral over the surface density profile; in kinematics, the line-of-sight velocity dispersion is a double integral over three-dimensional mass and light profiles. For most elliptical density profiles, we are unable to express these integrals with elementary or special functions. Special functions are numerically well-studied, hence fast algorithms to compute them are usually available. Thus, these functions can be numerically convenient for evaluating model-predicted observables in large numbers (e.g., $\sim10^6$) when sampling the lens model posterior with Markov chain Monte Carlo (MCMC) methods, or when only searching for the best-fit model even. Otherwise, it would be inefficient to numerically compute lensing and kinematic integrals for general elliptical profiles.

Usually, the numerical difficulties are circumvented through simplifying assumptions or approximations. We now outline common assumptions and approximations in kinematic and lensing analyses, noting that these can limit the accuracy and precision of the inferences.

Either axisymmetry or spherical symmetry is usually assumed for kinematic analysis. If we start with a surface density profile for lensing analysis, we need to deproject this profile along the line of sight to compute the kinematics. This deprojection has an infinite degeneracy \citep{Contopoulos56}. Therefore, it is necessary to choose a line-of-sight symmetry when deprojecting, for which either axisymmetry or spherical symmetry is often a convenient choice.

\begin{figure}
	\includegraphics[width=0.5\textwidth]{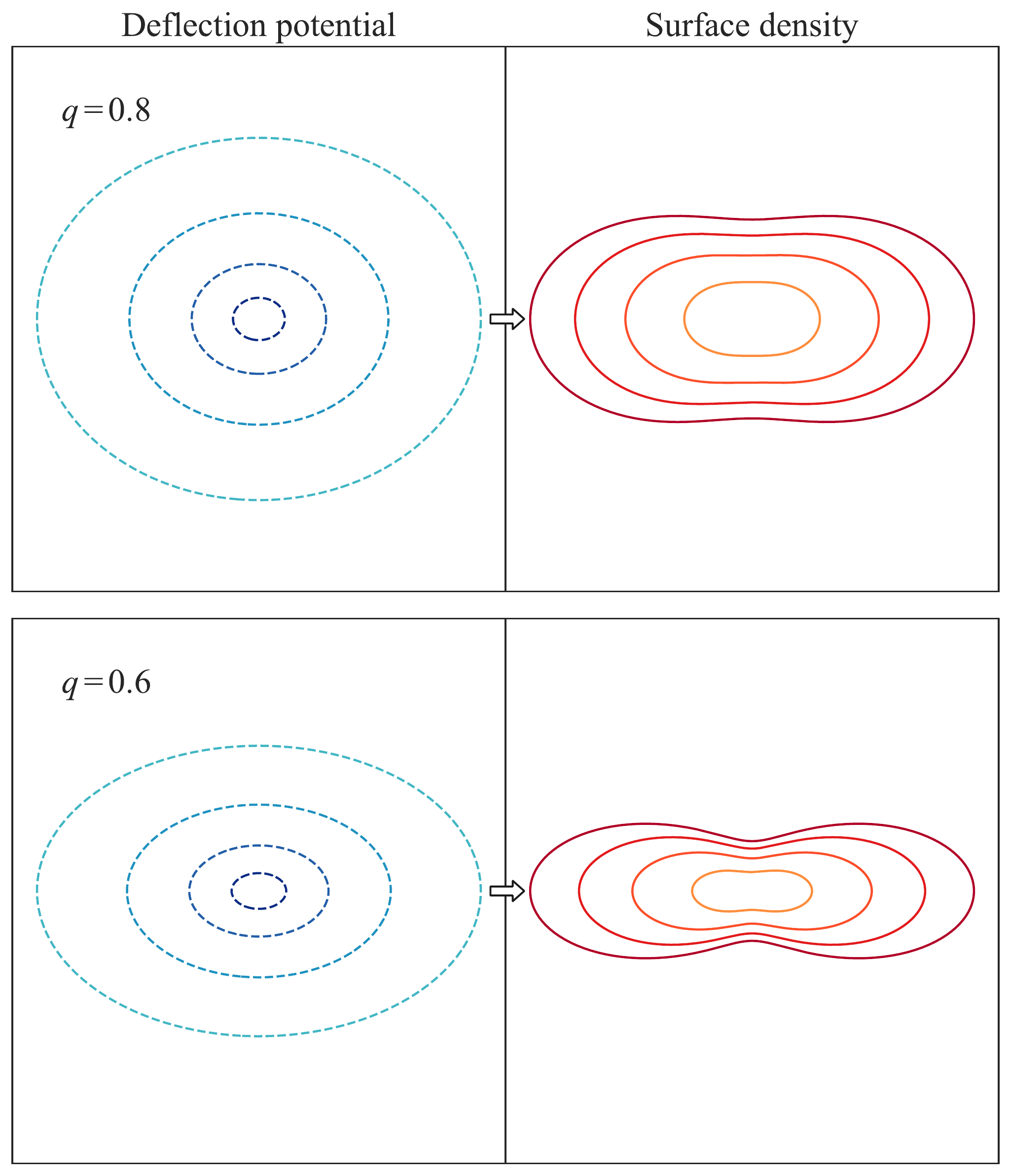}
	\caption{Elliptical deflection potential (left column) producing dumbbell-shaped surface-density (right column). The dashed contours in the left column are isopotential curves for S\'ersic profile with $n_{\rm \text{S\'ersic}}=4$, and the solid contours in the right column are corresponding isodensity curves. The axis ratios are $q=0.8$ in the top row, and $q=0.6$ in the bottom row. The dumbbell shape in surface density is unphysical and it gets more pronounced for higher ellipticity in the deflection potential. Hence, we can not use elliptical deflection potential to simplify lensing analysis of moderately elliptical galaxies. We need to treat ellipticity in the surface density, not deflection potential, to make our lensing analysis generally consistent with our physical priors.
	\label{fig:dumbbell_convergence}
	}
\end{figure}

In lensing analysis, spherical symmetry is rarely sufficient and we need to consider ellipticity to achieve the required precision. All the lensing quantities are related to the deflection potential or its derivatives. The time delay depends on the deflection potential difference. The gradient of the deflection potential gives the deflection angle. The Hessian of the deflection potential relates to the surface density, the shear, and the magnification. To efficiently compute these quantities for the elliptical case, we can find the following three approximations in the literature:
\begin{enumerate}[leftmargin=15pt,topsep=5pt, itemsep=5pt]
	\item \textit{Ellipticity in deflection potential:} The gradient and the Hessian of an elliptical deflection potential can be easily computed through numerical differentiation \citep[e.g.,][]{Kovner87b, Golse02}. However, this solution is not general, as the surface density becomes dumbbell-shaped for an elliptical deflection potential with axis ratio $q \lesssim 0.6$ (Fig. \ref{fig:dumbbell_convergence}, \citealt{Kassiola93}). This oddly shaped surface-density is unphysical.
	\item \textit{Elliptical power-law profile:} We can efficiently compute the lensing quantities for the elliptical power-law profile using numerical approximation or analytical expressions (for the isothermal case, \citealt{Kormann94}; and for the general case, \citealt{Barkana98}; \citealt{Tessore15}). This profile can be sufficient to use in statistical studies that do not require detailed modelling of individual lenses \citep[e.g.,][]{Koopmans09, Sonnenfeld13}. However, adopting a power-law profile artificially breaks the MSD and it could potentially bias the $H_0$ measurement \citep[e.g.,][]{Schneider13, Sonnenfeld18c}. Therefore, we need to explore different, physically-motivated mass models, such as a \textit{composite model} that explicitly accounts for the luminous and the dark components \citep{Suyu14, Yildirim19}.
	\item \textit{Chameleon profile:} The S\'ersic profile well describes the surface brightness of a galaxy \citep{Sersic68}. The Chameleon profile approximates the S\'ersic profile within a few per cent \citep{Dutton11}. We can efficiently compute lensing quantities for the Chameleon profile using analytic expressions. However, this profile only describes the baryonic component. The precise lensing analysis of elliptical Navarro--Frenk--White (NFW) profile for the dark component still lacks a general solution \citep{Navarro97}.
\end{enumerate}
In a nutshell, these approximations for elliptical lensing analysis are only applicable in restricted regimes as described above.

In this paper, we present a general method to precisely compute the gradient and the Hessian of the deflection potential for any elliptical surface-density profile. The method follows a ``divide and conquer'' strategy. We can approximately divide, or decompose, an elliptical profile into Gaussian components \citep[e.g.,][]{Bendinelli91}. For this Gaussian decomposition, we devise a fast and accurate algorithm by introducing an integral transform. For each Gaussian component, we derive analytic expressions of the gradient and the Hessian of the deflection potential. For the deprojected Gaussian component, \citet{Cappellari08} derives the line-of-sight velocity dispersion. We can combine the computed quantities from each Gaussian component back together to obtain these quantities for the total density profile. In this way, the lensing and kinematic descriptions are self-consistently unified. At the same time, this method is general, as we can apply it to lensing with any elliptical surface-density and to kinematics with either axisymmetry or spherical symmetry. Our method is more efficient than numerical integration to compute lensing quantities for an elliptical surface-density profile.

We organize this paper as follows. In Section \ref{sec:divide_and_conquer}, we motivate the ``divide and conquer'' strategy behind our method and introduce an integral transform that provides a fast algorithm for decomposing any elliptical surface-density profile into Gaussian components. In Section \ref{sec:gaussian_kinematics}, we summarize the kinematic description of the Gaussian components from \citet{Cappellari08}. In Section \ref{sec:lensing}, we derive the gradient and the Hessian of the deflection potential for elliptical Gaussian surface-density. Next in Section \ref{sec:proof_of_concept}, we demonstrate a proof-of-concept for our method using simulated data. Then, we summarize the paper in Section \ref{sec:conclusion}. Additionally in Appendix \ref{app:theorems}, we prove some fundamental theorems for the integral transform introduced in Section \ref{sec:divide_and_conquer}. This integral transform with a Gaussian kernel is the continuous case of Gaussian decomposition. The theorems in Appendix \ref{app:theorems} establish the existence, uniqueness, and invertibility of this transform.

\section{Decomposing an elliptical profile into Gaussian components} \label{sec:divide_and_conquer}

\begin{figure*}
	\includegraphics[width=\textwidth]{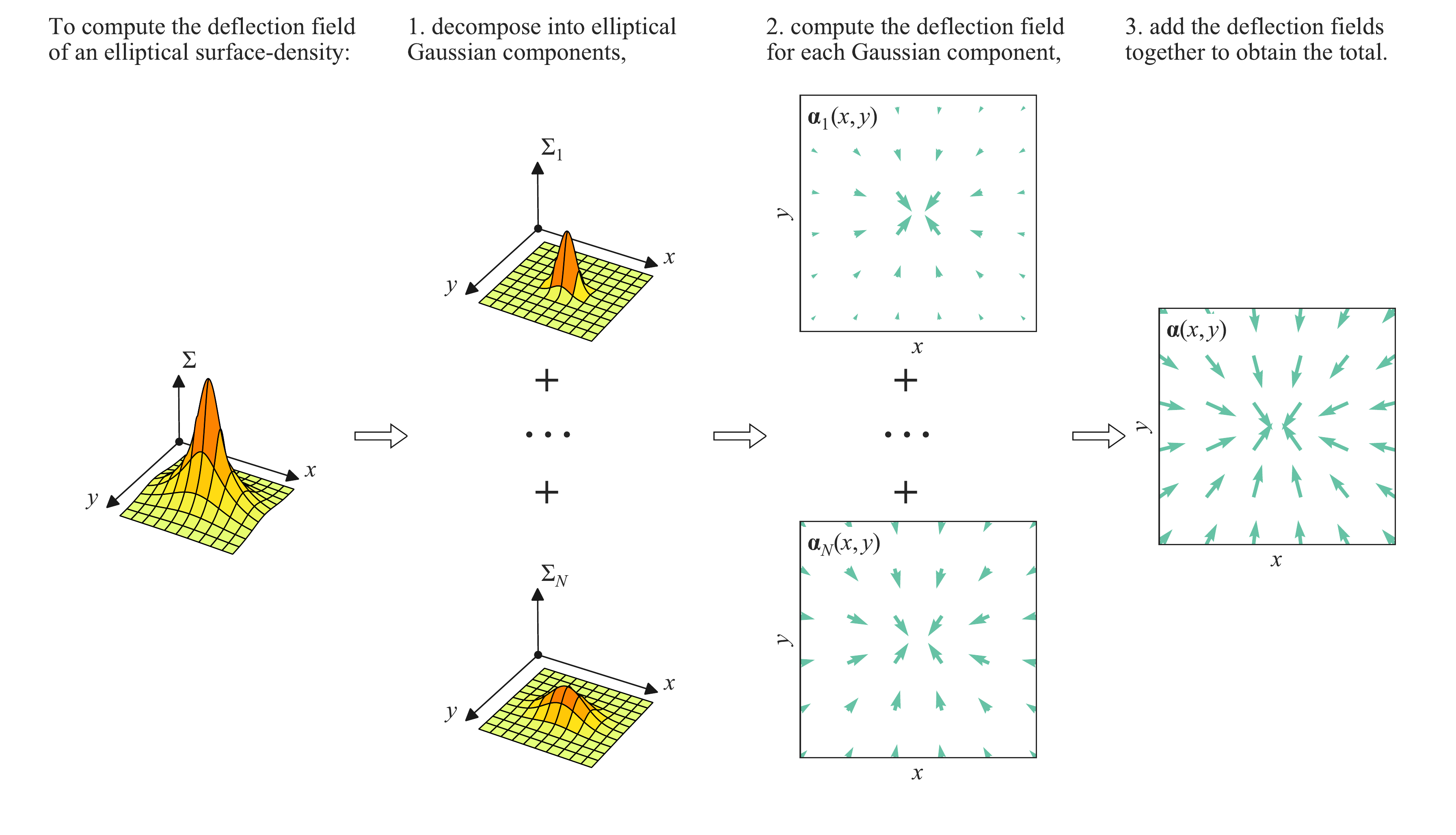}
	\caption{``Divide and conquer'' strategy to compute lensing quantities for elliptical surface-density profile. In this figure, we choose the deflection field as the quantity of interest to illustrate the method. However, this method works equally well for other quantities such as the lensing shear and the line-of-sight velocity dispersion. Each column demonstrates one step in the strategy and the arrows show the progression of these steps. We explain each step in the text at the top of the corresponding column.
	\label{fig:divide_conquer}
	}
\end{figure*}

To make the computation tractable, we aim to decompose an elliptical surface-density profile into simpler functions. This function should be simple enough so that we can both 
\begin{enumerate}[leftmargin=20pt,topsep=0pt]
	\item express the deflection angle in terms of elementary or special functions, and
	\item easily deproject it into three-dimension and compute the enclosed mass for kinematic analysis.
\end{enumerate}
The Gaussian function meets both of these criteria. We validate criterion (i) in Section \ref{sec:gaussian_lensing}, where we express the deflection angle for an elliptical Gaussian surface-density profile with the complex error function. For criterion (ii), deprojecting a two-dimensional Gaussian into three-dimension is straightforward, as the Abel inversion of a two-dimensional Gaussian is a three-dimensional Gaussian. The enclosed mass for a three-dimensional Gaussian has the form of the error function, which we can efficiently compute without integrating numerically. Given these points, we approximately decompose an elliptical surface-density profile as
\begin{equation} \nonumberoff
	\Sigma(x, y) \approx \sum_{j=0}^J {\Sigma_0}_{j} \exp\left(-\frac{q^2x^2+y^2}{2 \sigma_j^2}\right),	
\end{equation}
where ${\Sigma_0}_ j$ is the amplitude of the $j$-th Gaussian, and all the components have a common axis ratio $q$. Similar decomposition into Gaussian components has been used in the literature to fit the surface brightness profile of galaxies [called as the multi-Gaussian expansion (MGE) by \citet{Emsellem94}, and the mixture-of-Gaussians by \citet{Hogg13}]. The lensing and kinematic quantities of our interest -- namely the gradient and the Hessian of the deflection potential, and the line-of-sight velocity dispersion -- follow the principle of superposition. As a result, we can compute these quantities separately for each Gaussian component and then add them together to recover these quantities for the total surface-density profile (see Fig. \ref{fig:divide_conquer}).

We describe the kinematics and lensing analyses for the Gaussian components in Sections \ref{sec:gaussian_kinematics} and \ref{sec:lensing}, but first, we need a fast method to decompose an elliptical profile into Gaussian components. \citet{Cappellari02} presents a method that uses non-linear optimization. However, this non-linear optimization method is computationally too expensive to implement within MCMC. Although, this method has been implemented to compute the kinematic observable while sampling from the lens model posterior \citep[e.g.,][]{Birrer19}. Computing the kinematic observable can involve at least one numerical integration, which is the main bottleneck in the efficiency, not the non-linear optimization method for Gaussian decomposition. To make our lens-modelling method efficient, we require a \textbf{(i)} general, \textbf{(ii)} precise, and \textbf{(iii)} fast technique to decompose a function into Gaussians. In Sections {\ref{sec:integral_transform} and \ref{sec:2d_decomposition}, we provide a technique that satisfies these three requirements.

\subsection{An integral transform for fast Gaussian decomposition} \label{sec:integral_transform}

Now, we introduce an integral transform with a Gaussian kernel. Using this transform, we obtain an algorithm to efficiently decompose an elliptical surface-density profile into Gaussian components.

We start with the simple one-dimensional case of the integral transform. We aim to approximate a function $F(x)$ as a sum of Gaussian components as
\begin{equation} \label{eq:pre_integral_transform_approx}
	F(x) \approx \sum_{n=0}^{N} A_n \exp \left( -\frac{x^2}{2\sigma_n^2} \right),
\end{equation}
where $A_n$ and $\sigma_n$ are respectively the amplitude and the standard deviation of the $n$-th Gaussian component. We can convert this discrete summation into a continuous integral by taking $N \to \infty$. Accordingly, we define the following integral transform:
\begin{equation} \label{eq:integral_transform_in_body}
	F(x) \equiv \int_0^\infty \frac{f(\sigma)}{\sqrt{2 \uppi} \sigma} \exp\left( -\frac{x^2}{2\sigma^2}\right) \dd \sigma.
\end{equation}
Here, the amplitude $A_n$ is converted into a function $f(\sigma)/\sqrt{2\uppi}\sigma$. We call $F(x)$ as the \textit{transform} of $f(\sigma)$. We prove three fundamental properties of this integral transform in Appendix \ref{app:theorems}. These three properties tell us that
\begin{enumerate}[leftmargin=20pt,topsep=0pt, itemsep=0pt]
	\item this integral transform exists for most mass and light profiles of practical use,
	\item the transform is unique for these functions, and
	\item the integral transform is invertible.
\end{enumerate}
We call $f(\sigma)$ as the \textit{inverse transform} of $F(x)$. The inverse transform is given by
\begin{equation} \label{eq:inverse_transform_in_body}
	f(\sigma) = \frac{1}{\mathrm{i} \sigma^2} \sqrt{\frac{2}{\uppi}} \int_{C} z F(z) \exp \left( \frac{z^2}{2 \sigma^2} \right) \dd z.
\end{equation}
Here, $\mathrm{i}$ is the imaginary unit as $\mathrm{i} = \sqrt{-1}$. Also, we have extended $F(x)$ to some region on the complex plane and wrote it as $F(z)$, where $z$ is a complex variable. The contour $C$ for the integral lies within the region where $F(z)$ is defined (for details, see Appendix \ref{app:theorems} and Fig. \ref{fig:roc}). In Section \ref{sec:inverse_transform_algorithm}, we provide an algorithm that does not require $C$ to be explicitly specified for computing the inverse transform $f(\sigma)$.

We can use the inverse transform to decompose a function into Gaussian components and the forward transform to recover the original function by combining the Gaussian components. We first provide an efficient algorithm to compute the inverse transform from equation (\ref{eq:inverse_transform_in_body}), then we discuss a method for computing the forward transform from equation (\ref{eq:integral_transform_in_body}).

\subsubsection{Computing the inverse transform}\label{sec:inverse_transform_algorithm}

The integral transform in equation (\ref{eq:integral_transform_in_body}) can be converted into a Laplace transform by suitable change of variables (Remark \ref{rmk:laplace}). Therefore, we can use any of the several algorithms available for inverse Laplace transform by appropriately changing the variables \citep[for a simple overview of the algorithms, see][]{Abate06}. In this paper, we modify the Euler algorithm to approximate equation (\ref{eq:inverse_transform_in_body}) as
\begin{equation} \label{eq:inverse_num_approx}
	f(\sigma) \approx \sum_{n=0}^{2P}  \eta_n \re \left[ F \left( \sigma \chi_n \right) \right].
\end{equation}
\citep{Abate00}. Here, the weights $\eta_n$ and nodes $\chi_n$ can be complex-valued and they are independent of $f(\sigma)$. The weights and the nodes are given by
\begin{equation}
 \begin{split}       
 &\chi_n = \left[\frac{2P \log (10)}{3} + 2\uppi \mathrm{i} n \right]^{1/2}, \\
 &\eta_n = (-1)^n\ 2\sqrt{2\uppi}\ 10^{P/3} \xi_p, \\
 &\xi_0 = \frac{1}{2},\quad \xi_n = 1,\ 1\leq n \leq P,\quad \xi_{2P} = \frac{1}{2^P}, \\
 &\xi_{2P-n} = \xi_{2P-n+1} + 2^{-P} \binom Pn,\ 0 < n < P. 
 \end{split}
\end{equation}
 We can precompute the weights and the nodes just once before the MCMC sampling. In that way, computing them does not add any extra burden in computing the likelihood. The precision of the inverse transform is $\sim \mathcal{O}(10^{-0.6P})$ \citep{Abate06}. Therefore, the value of $P$ can be appropriately chosen to achieve a required precision. Note that the decimal precision of the machine sets an effective upper limit for $P$. For example, the precision will not improve with increasing $P$ when $P \gtrsim 12$ for 32-bit floating point number, and when $P \gtrsim 27$ for 64-bit floating point number. Thus, equation (\ref{eq:inverse_num_approx}) gives a straightforward, fast, and precise algorithm to compute the inverse transform.

\subsubsection{Computing the forward transform} \label{subsubsec:forward_transform}
Let us approximate the forward transform integral such that we can recover $F(x)$ from only a finite number of $f(\sigma)$ values computed at fixed $\sigma$'s. This finite number should be on the order of tens to keep the lensing analysis manageable, as we have to compute lensing quantities for each Gaussian component individually. We write equation (\ref{eq:integral_transform_in_body}) as
\begin{equation} \label{eq:gaussian_expansion_1d}
\begin{split}
	F(x) &= \frac{1}{\sqrt{2\uppi}} \int_0^{\infty} f(\sigma) \exp \left(-\frac{x^2}{2 \sigma^2} \right) \dd (\log \sigma) \\
	\Rightarrow F(x) &= \lim_{N \to \infty} \sum_{n=1}^N \frac{f(\sigma_n)}{\sqrt{2\uppi}} \exp \left(-\frac{x^2}{2\sigma_n^2} \right) \Delta(\log\sigma)_n \\
	\Rightarrow F(x) &\approx \sum_{n=1}^{N} A_n \exp \left(-\frac{x^2}{2\sigma_n^2} \right).
\end{split}
\end{equation}
We have recovered the form of equation (\ref{eq:pre_integral_transform_approx}) by taking logarithmically spaced $\sigma_n$. Here, the amplitudes are $A_n = w_n f(\sigma_n) \Delta (\log \sigma)_n/\sqrt{2\uppi}$. The weights $w_n$ depend on the choice of the numerical integration method. We use the trapezoidal method with weights $w_1 = 0.5$, $w_n = 1$ for $1<n<N$, $w_N=0.5$, as the trapezoidal method is highly efficient to numerically compute integrals of this form \citep{Goodwin49}. As a result, equation (\ref{eq:gaussian_expansion_1d}) efficiently recovers $F(x)$ from only a finite number of $f(\sigma)$ values. 

\begin{figure*}
\includegraphics[width=\textwidth]{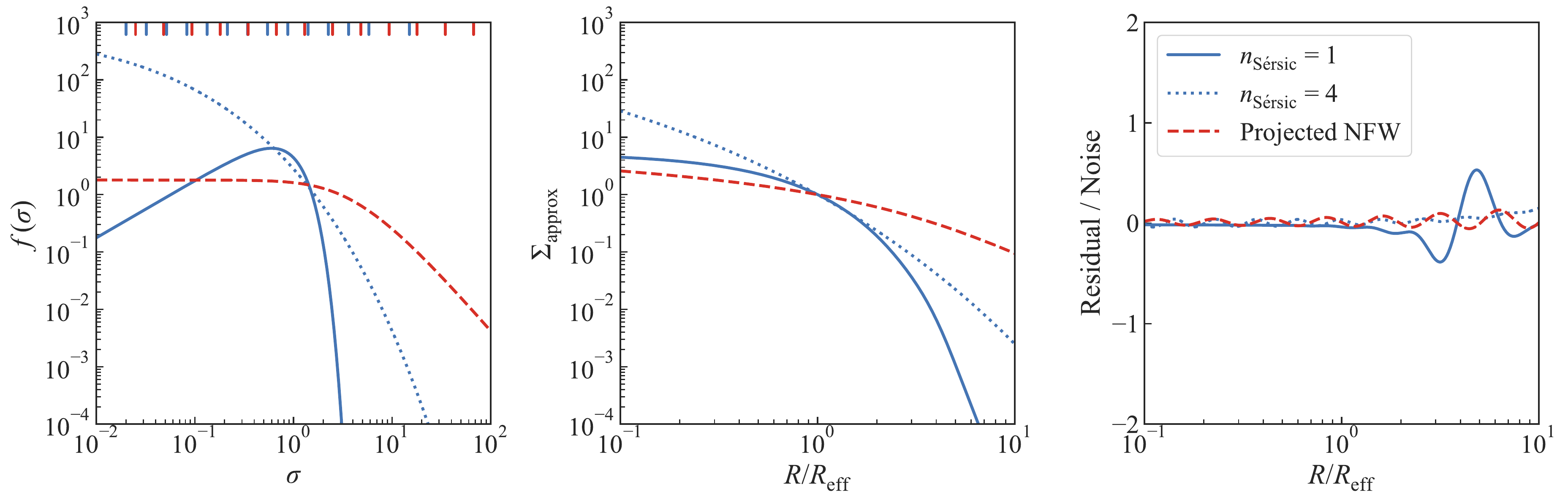}	
\caption{Decomposing the S\'ersic profile and the projected NFW profile into Gaussian components using the integral transform with a Gaussian kernel. The blue lines correspond to the S\'ersic profiles with: solid for S\'ersic index $n_{\rm \text{S\'ersic}} = 1$ and dotted for $n_{\rm \text{S\'ersic}}=4$. The red, dashed lines correspond to the two-dimensional projected NFW profile. \textbf{Left:} the inverse transform $f(\sigma)$ of the S\'ersic profiles and the projected NFW profile. Here, we choose the NFW scale radius $r_{\rm s} = 5R_{\rm eff}$. To decompose a function into 15 Gaussian components, we only need to compute $f(\sigma)$ at 15 points. These points are marked along the top border as blue ticks for the S\'ersic profiles and as red ticks for the NFW profile. \textbf{Center:} recovering the original profile as $\Sigma_{\rm approx}(R)$ using the forward transform by combining the 15 Gaussian components. We do not plot the true form of $\Sigma_{\rm \text{S\'ersic}}(R)$ or $\Sigma_{\rm \text{NFW}}(R)$ because they are visually almost indistinguishable from $\Sigma_{\rm approx}(R)$. \textbf{Right:} noise-normalized difference between the recovered profile $\Sigma_{\rm approx}(R)$ and the true form of $\Sigma_{\rm \text{S\'ersic}}(R)$ or $\Sigma_{\rm \text{NFW}}(R)$. We assume 1 per cent Poisson noise at $R = R_{\rm eff}$ to obtain the noise level for normalizing the residual. Our method approximates the NFW profile and the S\'ersic profile as a sum of 15 Gaussians within the noise level for $0.1R_{\rm eff} \leq R \leq 10R_{\rm eff}$.
	\label{fig:1d_example}
}
\end{figure*}

\begin{figure*}
\includegraphics[width=\textwidth]{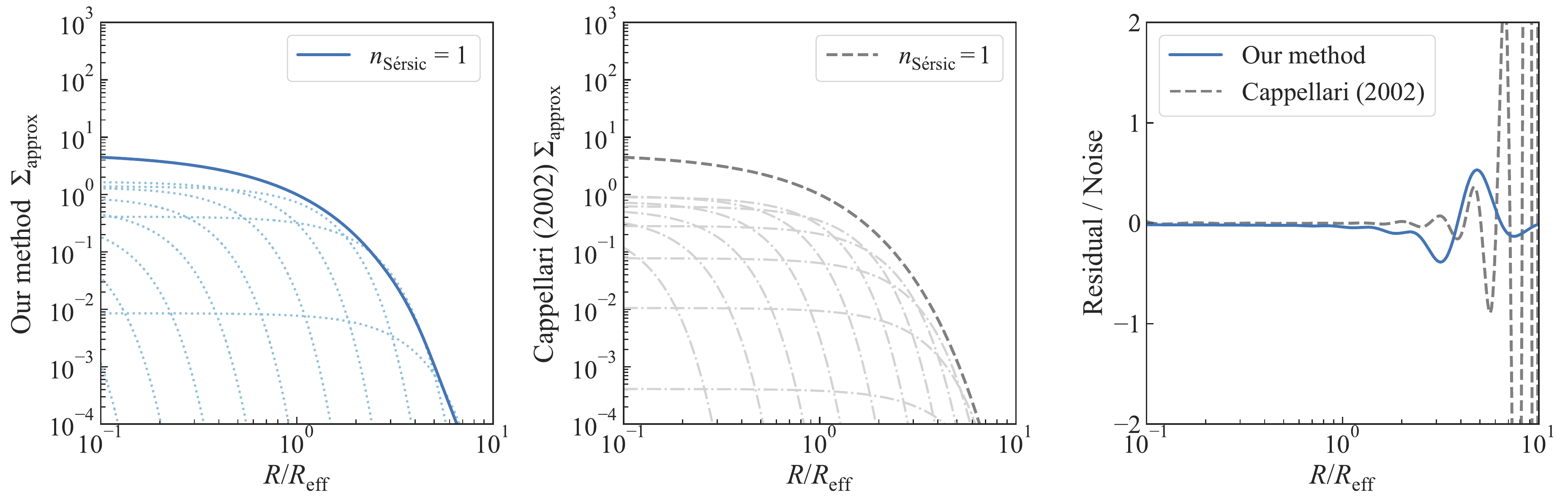}	
\caption{Comparison between our method and the multi-Gaussian expansion (MGE) method from \citet{Cappellari02} to decompose a one-dimensional function into Gaussian components. Here, we only show the case for a S\'ersic function with $n_{\rm \text{S\'ersic}}=1$, however the cases for higher S\'ersic indices of for the projected NFW profile are qualitatively similar or better. \textbf{Left:} the S\'ersic function (solid, blue line) approximated with 15 Gaussian components using our Gaussian decomposition method. The dotted, lighter-blue lines show the individual Gaussian components. Some of the Gaussian components are out of the figure range. \textbf{Center:} same as the left figure but using the MGE method with 15 Gaussian components. The dashed, grey line shows the S\'ersic function approximated by MGE and the dot-dashed, lighter-grey lines show individual Gaussian components. \textbf{Right:} comparison of the noise-normalized residual for the two methods. We assume 1 per cent Poisson noise at effective radius $R_{\rm eff}$ to obtain the noise level for normalizing the residual. The MGE method approximates the S\'ersic function within the noise level up to $\sim 6 R_{\rm eff}$, whereas our method approximates the S\'ersic function within the noise level up to $10R_{\rm eff}$. More importantly, our method is $\sim 10^3$ times faster than the MGE method to decompose a one-dimensional function into Gaussian components.
	\label{fig:1d_example_compare}
	}
\end{figure*}

As an example, we demonstrate the integral transform method to decompose the NFW profile and the S\'ersic profile into Gaussian components (Fig. \ref{fig:1d_example}).  The two-dimensional projected NFW profile is given by
\renewcommand*{\arraystretch}{1.7}
\begin{equation} \label{eq:2d_nfw}
	\Sigma_{\rm NFW} (R) = \frac{\rho_{\rm s} r_{\rm s}}{(R/r_{\rm s})^2 - 1} \times
	\left\{
		\begin{array} {ll}
			1 - \frac{ \sec^{-1} \left(R/r_{\rm s}\right)}{\sqrt{(R/r_{\rm s})^2-1}}  & (R > r_{\rm s}), \\
			\frac{1}{3} \left[(R/r_{\rm s})^2 - 1\right] & (R = r_{\rm s}), \\
			1- \frac{ {\rm sech}^{-1} {(R/r_{\rm s}) } } {\sqrt{1-(R/r_{\rm s})^2}}  & (R < r_{\rm s})
		\end{array}
	\right.
\end{equation}
\citep{Bartelmann96}. Here, $\rho_{\rm s}$ is the three-dimensional density normalization, and $r_{\rm s}$ is the scale radius. The S\'ersic profile is given by
\begin{equation} \nonumberoff
	\Sigma_{\rm \text{S\'ersic}}(R) = \Sigma_{\rm eff} \exp\left[ - b_n \left\{ (R/R_{\rm eff})^{1/n_{\rm \text{S\'ersic}}} - 1 \right\} \right]
\end{equation}
\citep{Sersic68}. Here, the normalizing factor $b_n$ ensures that half of the total projected mass is contained within the effective radius $R_{\rm eff}$.  We only need to compute as many $f(\sigma)$ values as the number of Gaussian components. We can appropriately choose this number to achieve the required precision for approximating the original function within a given range of $R$. In this example, we set $r_{\rm s}=5R_{\rm eff}$ and assume 1 per cent Poisson noise at $R = R_{\rm eff}$. Then, we can approximate both the projected NFW profile and the S\'ersic function within the noise level with only 15 Gaussian components in the range $0.1R_{\rm eff} \leq R \leq 10R_{\rm eff}$. The standard deviations $\sigma_n$ of the 15 Gaussians are logarithmically spaced between $0.005r_{\rm s}$ and $50r_{\rm s}$ for the NFW profile, and between $0.02R_{\rm eff}$ and $15R_{\rm eff}$ for the S\'ersic profile. Thus using the integral transform method, we can decompose a function into Gaussian components within any required precision by appropriately choosing the component number $N$.

\subsubsection{The integral transform method is more efficient than the MGE method.}
We compare our method to decompose a one-dimensional function into Gaussian components with the MGE method (Fig. \ref{fig:1d_example_compare}). We use both methods to decompose the S\'ersic function into 15 Gaussian components. The MGE method approximates the S\'ersic function within the noise level up to $\sim 6R_{\rm eff}$, whereas our method approximates the S\'ersic function within the noise level up to 10$R_{\rm eff}$ with the same number of components. Albeit, we can increase the number of Gaussian components in the MGE method to reach the desired precision within a given radius. In our method, the precision of the decomposition can be affected by both $P$ in equation (\ref{eq:inverse_num_approx}) and the number of Gaussians $N$ in equation (\ref{eq:gaussian_expansion_1d}). However, if $P$ is appropriately chosen so that $10^{-0.6P}$ is sufficiently (e.g., by a factor of $10^{-2}$--$10^{-4}$) smaller than the required precision, then the precision predominantly depends on $N$. For lensing and kinematic analyses, increasing the number of Gaussian components $N$ introduces more computational burden than increasing $P$. Therefore, it is advisable to first choose a sufficiently large $P$ and then adjust the number of Gaussians to achieve the required precision. Note that the real power of our method is in its efficiency. A \textsc{python} implementation of our method is $\sim 10^3$ times faster than the MGE method to decompose a one-dimensional function into Gaussian components with similar or better precision.

\subsection{Decomposing a two-dimensional elliptical profile with the one-dimensional transform}\label{sec:2d_decomposition}

So far, we have discussed the one-dimensional case of the integral transform; now we show that the one-dimensional transform is sufficient to decompose a two-dimensional elliptical profile. We can extend the one-dimensional integral transform from equation (\ref{eq:integral_transform_in_body}) into a two-dimensional integral transform for a function $f(\sigma_1, \sigma_2)$ as
\begin{equation} \label{eq:2d_transform}
\begin{split}
	F(x, y) =  \int_0^\infty \dd \sigma_1 \int_0^\infty \dd \sigma_2 \  \frac{f(\sigma_1, \sigma_2)}{2 \uppi \sigma_1 \sigma_2}\ \exp\left( -\frac{x^2}{2\sigma_1^2} - \frac{y^2}{2\sigma_2^2} \right).
\end{split}
\end{equation}
If $F(x,y)$ is elliptically symmetric, then we can express it as $F(R)$ in terms of the \textit{elliptical radius} $R=\sqrt{q^2x^2+y^2}$ and axis ratio $q$. Then, we can write
\begin{equation}\nonumberoff
	\begin{split}
		F(x,y) &= F\left(R(q)\right) = \int_0^{\infty} F\left(R(\varrho)\right)\ \delta(\varrho-q)\ \dd \varrho \\
		&= \int_0^{\infty} \dd \varrho \ \delta(\varrho-q \mathcal{}) \int_0^{\infty} \dd \sigma\  \frac{f_y(\sigma)}{\sqrt{2\uppi}\sigma}\ \exp\left(-\frac{R(\varrho)^2}{2\sigma^2} \right),
	\end{split}
\end{equation}
where $f_y(\sigma)$ is the inverse transform of $F(0,y)$. If we make the change of variables $\sigma_1 = \sigma/\varrho$, $\sigma_2=\sigma$, this integral becomes
\begin{equation} \nonumberoff \label{eq:2d_transform_2}
	\begin{split}
		F(x,y) &= \frac{1}{2\uppi} \int_0^{\infty} \dd \sigma_1 \int_0^{\infty} \dd \sigma_2 \frac{\sqrt{2\uppi}q \ \delta(\sigma_2/\sigma_1-q)\ f_y(\sigma_2)}{\sigma_1 \sigma_2} \\
		&\qquad\qquad\qquad\qquad \times \exp \left(-\frac{x^2}{2\sigma_1^2}-\frac{y^2}{2\sigma_2^2} \right).
	\end{split}
\end{equation} 
Because of the uniqueness property (Theorem \ref{thm:uniqueness}), comparing equations (\ref{eq:2d_transform}) and (\ref{eq:2d_transform_2}) we can write
\begin{equation} \nonumberoff
	f(\sigma_1, \sigma_2) = 	\sqrt{2\uppi}q \ \delta\left(\frac{\sigma_2}{\sigma_1}-q \right)f_y(\sigma_2).
\end{equation}
Therefore, for an elliptically symmetric function, it is sufficient to numerically compute the one-dimensional inverse transform $f_y(\sigma_2)$ along the $y$-axis. As a result, we can express a two-dimensional elliptical function as a sum of elliptical Gaussian components as
\begin{equation} \label{eq:gauss_expansion_elliptical}
	F(x, y) \approx	\sum_{n=1}^{N} A_n \exp \left( -\frac{q^2 x^2 + y^2}{2\sigma_n^2} \right).
\end{equation}

By now, we have shown that the integral transform method meets all of our three requirements for decomposing a function into Gaussian components:
\begin{enumerate}[topsep=0pt, leftmargin=20pt]
	\item \textit{generality}: the method applies to most mass and light profiles of practical use,
	\item \textit{precision}: the method achieves any required precision over a given range by appropriately choosing the number of Gaussian components, and
	\item \textit{efficiency}: the method runs approximately $\sim 10^3$ times faster than the previously available method.
\end{enumerate}
With these three requirements met, our lensing analysis method has cleared the first hurdle to be feasible in practice. In Section \ref{sec:gaussian_lensing}, we show that we can efficiently compute the gradient and the Hessian of the deflection potential for an elliptical Gaussian surface-density. With that, our method also clears the final hurdle to be efficient. Next in Section \ref{sec:gaussian_kinematics}, we summarize the kinematic analysis for the Gaussian components from \citet{Cappellari08}. Then in Section \ref{sec:lensing}, we describe the lensing analysis for the Gaussian components and complete the unification of lensing and kinematic descriptions.

\section{Kinematics of Gaussian components} \label{sec:gaussian_kinematics}

\citet{Cappellari08} presents the Jeans anisotropic modelling of kinematics for a mass profile decomposed into Gaussian components. We summarize the analysis here to complete our unified framework. The kinematic observable is the luminosity-weighted, line-of-sight velocity dispersion. The velocity dispersion can be an integrated measurement within a single aperture or it can be spatially resolved on the plane of the sky. To compute this quantity for a combination of mass and light profiles, we need to solve the Jeans equations. We can decompose the surface mass-density profile into Gaussian components as
\begin{equation} \label{eq:2d_mass_kin}
	\Sigma (x, y) \approx \sum_{j=1}^J	{\Sigma_{0}}_{j} \exp \left( - \frac{q_j^2 x^2 + y^2}{2 \sigma_j^2} \right),
\end{equation}
and decompose the surface brightness profile into Gaussian components as
\begin{equation} \label{eq:2d_light_kin}
	I (x, y) \approx \sum_{k=1}^K	{I_{0}}_{k} \exp \left( - \frac{q_k^2 x^2 + y^2}{2 \sigma_k^2} \right).
\end{equation}
Here, we use different subscript letters to make the context of the Gaussian decomposition clear: we use the subscript $j$ for a component of the mass profile and the subscript $k$ for a component of the light profile. We have also allowed different ellipticity for each Gaussian component represented by $q_j$ or $q_k$. This is the most general case, for example, when structures with different ellipticities constitute the total mass or light distribution. We first need to deproject these two-dimensional profiles into three-dimension as the kinematic quantities depend on the three-dimensional distributions of mass and light. We assume axisymmetry or spherical symmetry for the deprojected three-dimensional structure to circumvent the infinite degeneracy in deprojection. First, we provide the kinematic analysis for the axisymmetric case in Section \ref{sec:axisymmetric_kinematics}; then we do the same for the simpler case of spherical symmetry in Section \ref{sec:spherical_kinematics}.

\subsection{Axisymmetric case}\label{sec:axisymmetric_kinematics}

For an axisymmetric system, the cylindrical coordinates $(R^{\prime}, z^{\prime}, \phi^{\prime})$ are the most suitable to express the Jeans equations. We use the prime symbol to denote the coordinates in the system's \textit{symmetry-frame}, where the $z^{\prime}$-axis aligns with the axis of symmetry. We assign $(x, y, z)$ coordinates to the \textit{sky frame}, where the $z$-axis aligns with the line of sight and the $x$-axis aligns with the projected major axis. If the galaxy is inclined by an angle $\iota$, then the $(x^{\prime}, y^{\prime}, z^{\prime})$ coordinates in the symmetry frame relate to the the $(x, y, z)$ coordinates in the sky frame as
\renewcommand*{\arraystretch}{1.1}
\begin{equation} \nonumberoff
	\begin{pmatrix} x^{\prime} \\ y^{\prime} \\ z^{\prime}	\end{pmatrix} 
	= 
	\begin{pmatrix}
		1 &0 &0 \\ 0 &\cos \iota  &-\sin \iota \\ 0 &\sin \iota & \cos \iota 
	\end{pmatrix}
	\begin{pmatrix} x \\ y \\ z	\end{pmatrix} .
\end{equation}
In the symmetry frame, equation (\ref{eq:2d_mass_kin}) deprojects into the mass density profile $\rho$ as
\begin{equation}
	\rho(R^{\prime}, z^{\prime}) = \sum_{j=1}^J 	\frac{{q_j^{\prime}}^2 {\Sigma_{0}}_{j}}{\sqrt{2\uppi} \sigma_j q_j} \exp \left( - \frac{{q_j^{\prime}}^2 {R^{\prime}}^2 + {z^{\prime}}^2}{2 \sigma_j^2} \right),
\end{equation}
and equation (\ref{eq:2d_light_kin}) deprojects into the light density $l$ as 
\begin{equation}
	l (R^{\prime}, z^{\prime}) = \sum_{k=1}^K	\frac{{q_k^{\prime}}^2 {I_0}_{k}}{\sqrt{2\uppi} \sigma_k q_k} \exp \left( - \frac{{q_k^{\prime}}^2 {R^{\prime}}^2 + {z^{\prime}}^2}{2 \sigma_k^2} \right).
\end{equation}
Here, the intrinsic axis ratio $q^{\prime}$ relates to the projected axis ratio $q$ as
\begin{equation} \nonumberoff
	q^{\prime} = \frac{\sqrt{q^2 - \cos^2 \iota}}{\sin \iota}.	
\end{equation}
We can first solve the Jeans equations for these mass and light density profiles in the symmetry frame to get the intrinsic velocity dispersions and then integrate along the line of sight to obtain the line-of-sight velocity dispersion.
 
The Jeans equations to solve for an axisymmetric system are 
\begin{equation} \nonumberoff
\begin{split}
	&\frac{b\ l\ \overline{v_{z^{\prime}}^2} - l\ \overline{v_{\phi^{\prime}}^2}}{R^{\prime}} + \frac{\partial \left( b\ l\ \overline{v_{z^{\prime}}^2} \right)}{\partial R^{\prime}} = - l \frac{\partial \Phi}{\partial R^{\prime}}, \\
	&\frac{\partial \left( l\ \overline{v_{z^{\prime}}^2} \right)}{\partial z^{\prime}} = - l \frac{\partial \Phi}{\partial z^{\prime}}.
\end{split}
\end{equation}
Here,  the gravitational potential $\Phi$ relates to the three-dimensional mass density $\rho$ by $\nabla^2 \Phi = \rho$, and $b$ represents the anisotropy as in $\overline{v_{R^{\prime}}^2} = b \overline{v_{z^{\prime}}^2}$. We can let $b_k$ for each luminous Gaussian component have different values to approximate the luminosity-weighted anisotropy parameter as
\begin{equation} \nonumberoff
	\beta_{z^{\prime}} (R^{\prime}, z^{\prime}) \equiv 1 - \frac{\overline{v_{z^{\prime}}^2}}{\overline{v_{R^{\prime}}^2}} \approx 1 - \frac{\sum_k l_k}{\sum_k b_k l_k}
\end{equation}
\citep{Binney82, Cappellari08}. The line-of-sight second velocity moment for total mass profile obtained from solving the Jeans equations is given by
\begin{equation} \label{eq:axisymmetric_vlos}
\begin{split}
	\overline{v_{\rm los}^2}	 (x, y) =& 2 \sqrt{\uppi} G \int_0^1 \sum_{j=1}^J \sum_{k=1}^K \frac{{q_j^{\prime}}^3 {q_k^{\prime}}^2 {\Sigma_0}_j {I_0}_k u^2}{\sigma_j q_j \sigma_k q_k} \\
	& \times \frac{\sigma_k^2\left( \cos^2 \iota + b_k \sin^2 \iota \right) + \mathcal{D}x^2 \sin^2 \iota}{\left(1-\mathcal{C}u^2 \right)\sqrt{\left(\mathcal{A} + \mathcal{B} \cos^2 \iota \right) \left[1 - \left(1-{q_j^{\prime}}^2\right)u^2 \right]}} \\
	& \times \exp \left(-\mathcal{A} \left[x^2 + \frac{(\mathcal{A}+\mathcal{B})y^2}{\mathcal{A}+\mathcal{B}\cos^2 \iota} \right] \right) \dd u,
	\end{split}
\end{equation}
where $G$ is the gravitational constant and
\begin{equation} \nonumberoff
	\begin{split}
		\mathcal{A} &= \frac{1}{2} \left( \frac{u^2 {q_j^{\prime}}^2}{\sigma_j^2} + \frac{{q_k^{\prime}}^2}{\sigma_k^2} \right), \\
		\mathcal{B} &= \frac{1}{2}\left\{ \frac{1-{q_k^{\prime}}^2}{\sigma_k^2} + \frac{{q_j^{\prime}}^2\left(1-{q_j^{\prime}}^2\right)u^4}{\sigma_j^2 \left[ 1- \left(1-{q_j^{\prime}}^2\right) u^2\right] } \right\},\\
		\mathcal{C} &= 1 - {q_j^{\prime}}^2 - \frac{{q_j^{\prime}}^2\sigma_k^2}{\sigma_j^2},\\
		\mathcal{D} &= 1 - b_k{q_k^{\prime}}^2 - \left[(1-b_k)\mathcal{C} + (1-{q_j^{\prime}}^2)b_k \right] u^2
	\end{split}
\end{equation}
\citep{Cappellari08}. The line-of-sight velocity dispersion $\sigma_{\rm los}$ relates to the second velocity moment by $\overline{v^2_{\rm los}} = v_{\rm mean}^2 + \sigma^2_{\rm los}$, where $v_{\rm mean}$ is the stellar mean velocity.

\subsection{Spherical case} \label{sec:spherical_kinematics}
If we assume the system is spherically symmetric, then the spherical coordinates $(r, \phi, \theta)$ are the most suitable to express the Jeans equations. In this coordinate system, the mass density profile deprojected from equation (\ref{eq:2d_mass_kin}) takes the form 
\begin{equation}
	\rho(r) = \sum_{j=1}^J \frac{{\Sigma_{0}}_j }{\sqrt{2 \uppi} \sigma_j q_j} \exp \left( -\frac{r^2}{2 \sigma_j^2}\right),
\end{equation}
and the light density profile deprojected from equation (\ref{eq:2d_light_kin}) turns into
\begin{equation}
	l(r) = \sum_{k=1}^K \frac{{I_{0}}_k}{\sqrt{2 \uppi} \sigma_k q_k} \exp \left( -\frac{r^2}{2 \sigma_k^2}\right).
\end{equation}
The projected axis ratio $q$ shows up in these equations to keep the total mass and luminosity conserved. We can express the three-dimensional enclosed mass for this density profile as
\begin{equation}
	M(r) = \sum_{j=1}^J \frac{2\uppi \sigma_j^2 {\Sigma_0}_j}{q_j} \left[ \erf \left( \frac{r}{\sqrt{2}\sigma_j} \right) - \sqrt{\frac{2}{\uppi}} \frac{r}{\sigma_j} \exp \left( - \frac{r^2}{2 \sigma_j^2} \right) \right],
\end{equation}
where $\erf\ (x)$ is the error function. The spherical Jeans equation is
\begin{equation} \nonumberoff
	\frac{\dd \left( l\ \overline{v_r^2}\right)}{\dd r} + \frac{2 \beta\ l\  \overline{v_r^2}}{r} = - l \frac{\dd \Phi}{\dd r},	
\end{equation}
where $\beta(r)$ is the anisotropy parameter given by 
\begin{equation}
	\beta(r) = 1 - {\overline{v_{\theta}^2}}/{\overline{v_{r}^2}}.
\end{equation}
Spherical symmetry imposes that  $\overline{v_{\theta}^2} = \overline{v_{\phi}^2}$. By solving the Jeans equation for the spherically symmetric case, we can obtain the line-of-sight second velocity moment as
\begin{equation} \label{eq:spherical_vlos}
	\overline{v_{\rm los}^2} (x, y) =  \frac{2G}{I(x, y)} \int_{\sqrt{x^2+y^2}}^\infty \mathcal{K}_\beta \left(\frac{r}{\sqrt{x^2+y^2}} \right)\ l(r)\ M(r)\ \frac{\dd r}{r}
\end{equation}
\citep{Mamon05}. Here, the function $\mathcal{K}_\beta (\upsilon)$ depends on the form of the anisotropy parameter $\beta(r)$. For the isotropic case with $\beta = 0$, the function $\mathcal{K}_\beta$ shapes into
\begin{equation}
	\mathcal{K}_{\beta} (\upsilon) = \sqrt{1 - \frac{1}{\upsilon^2}}.
\end{equation}
For the Osipkov-Merritt parameterization $\beta(r) = r^2/(r^2 + r_{\rm ani}^2)$, where $r_{\rm ani}$ is a scale radius, the function $\mathcal{K}_{\beta}$ takes the form
\begin{equation}
\begin{split}
	\mathcal{K}_{\beta} &\left( \upsilon \right) = \frac{\upsilon_{\rm ani}^2 + 1/2}{(\upsilon_{\rm ani}+1)^{3/2}} \left( \frac{\upsilon^2 + \upsilon^2_{\rm ani}}{\upsilon} \right) \tan^{-1} \left(\sqrt{\frac{\upsilon^2 - 1}{\upsilon_{\rm ani}^2 + 1}}\right) \\
	& \qquad\qquad\qquad\qquad - \frac{1/2}{\upsilon_{\rm ani}^2 + 1} \sqrt{1 - \frac{1}{\upsilon^2}}.
\end{split}
\end{equation}
with $\upsilon_{\rm ani} = r_{\rm ani}/\sqrt{x^2 + y^2}$  \citep{Osipkov79, Merritt85b, Merritt85}. See equation (A16) of \citet{Mamon05} for the form of $\mathcal{K}_{\beta}$ corresponding to other parameterizations of $\beta(r)$. When assuming spherical symmetry is sufficient, we can use equation (\ref{eq:spherical_vlos}) to compute the line-of-sight velocity dispersion in a much simpler way than the axisymmetric case [cf. equation (\ref{eq:axisymmetric_vlos})].

The kinematic description of an elliptical mass distribution by decomposing it into Gaussian components is thus well developed in the literature. In the next section, we unify the lensing description with the kinematic description under the same framework.

\section{Lensing by Gaussian components} \label{sec:lensing}

In this section, we present the lensing analysis for an elliptical surface-density profile decomposed into Gaussian components. In Section \ref{sec:divide_and_conquer}, we introduced an integral transform that efficiently decomposes an elliptical surface-density profile into Gaussian components as
\begin{equation} \nonumberoff
	\Sigma (x, y) \approx \sum_{j=1}^J	{\Sigma_{0}}_{j} \exp \left( - \frac{q_j^2 x^2 + y^2}{2 \sigma_j^2} \right)	.
\end{equation}
We can compute a lensing quantity for each individual Gaussian component, and then linearly add the contributions from all the components to obtain the total lensing quantity. For example, if $\balpha_j(x, y)$ is the deflection at position $(x, y)$ for the $j$-th Gaussian component, then the total deflection is simply given by $\balpha (x, y) = \sum_{j=1}^J \balpha_j (x, y)$. Therefore, it is sufficient to analyze the lensing properties of one elliptical Gaussian component. We use the complex formulation of lensing to solve the deflection integral for an elliptical Gaussian surface-density profile. Below, we first lay out the complex formalism of lensing in Section \ref{sec:complex_lensing}; then we study the lensing properties of an elliptical Gaussian surface-density profile in Section \ref{sec:gaussian_lensing}.

\subsection{Complex formulation of lensing} \label{sec:complex_lensing}
The strong lensing effect is usually described using the vector formulation on the two-dimensional image plane. We first define the lensing quantities in the familiar vector formulation, then we translate them to the complex formulation. The convergence $\kappa$ is a dimensionless surface-density defined as $\kappa \equiv \Sigma/\Sigma_{\rm crit}$, where the critical density $\Sigma_{\rm crit}$ is given by
\begin{equation} \nonumberoff
	\Sigma_{\rm crit} = \frac{c^2 D_{\rm s}}{4 \uppi G D_{\rm ds} D_{\rm d}}.	
\end{equation}
Here, $c$ is the speed of light. The three angular diameter distances are $D_{\rm d}$: between the observer and the deflector, $D_{\rm s}$: between the observer and the source, and $D_{\rm ds}$: between the deflector and the source. The convergence $\kappa$ relates to the vector deflection angle ${\balpha}$ as $\kappa = \nabla \cdot \balpha / 2$. The deflection angle $\balpha$ is the gradient of the deflection potential as $\balpha = \nabla \psi$, thus the convergence $\kappa$ relates to the deflection potential $\psi$ as $\kappa = \nabla^2 \psi / 2$. The Hessian of the deflection potential is
\renewcommand*{\arraystretch}{2.2}
\begin{equation} \nonumberoff
	\bm{\mathsf{H}} = \begin{pmatrix}
 			\dfrac{\partial^2 \psi}{\partial^2 x} &\dfrac{\partial^2 \psi}{\partial x \partial y} \\
 			\dfrac{\partial^2 \psi}{\partial x \partial y} &\dfrac{\partial^2 \psi}{\partial^2 y}
 		\end{pmatrix}.
\end{equation}
The convergence $\kappa$, the shear parameters $(\gamma_1,\ \gamma_2)$, and the magnification $\mu$ relate to the Hessian, since we can express them as
\begin{equation}\nonumberoff
	\begin{split}
		\kappa &= \frac{1}{2}\left(\frac{\partial^2 \psi}{\partial^2 x} + \frac{\partial^2 \psi}{\partial^2 y} \right), \\
		\gamma_1 &= \frac{1}{2} \left(\frac{\partial^2 \psi}{\partial^2 x} - \frac{\partial^2 \psi}{\partial^2 y} \right), \\
		\gamma_2 &= \frac{\partial^2 \psi}{\partial x \partial y}, \\
		\mu &= \frac{1}{\det\ (\bm{\mathsf{I}} - \bm{\mathsf{H}})},
	\end{split}	
\end{equation}
where $\bm{\mathsf{I}}$ is the identity matrix. Therefore, if we start with a convergence $\kappa$ and derive the deflection $\balpha$ and the shear parameters $(\gamma_1,\ \gamma_2)$, then we can obtain the gradient and the Hessian of the deflection potential from them.

Now we reformulate the lensing quantities on the complex plane. Following \citet{Bourassa73}, we can express the deflection vector $\balpha$ as a complex quantity
\begin{equation} \nonumberoff
	\alpha(z) \equiv \alpha_x + \mathrm{i} \alpha_y,	
\end{equation}
where the complex quantity $z = x + \mathrm{i}y$ corresponds to the position vector $\bold{r} = (x,\ y)$. We can define a complex deflection potential $\psi(z)$ with its real part equal to the usual deflection potential \citep{Schramm90}. Then, the complex deflection angle is the Wirtinger derivative of the deflection potential as
\begin{equation} \nonumberoff
	\alpha (z) = \frac{\partial \psi}{\partial x} + \mathrm{i} \frac{\partial \psi}{\partial y} = 2 \frac{\partial \psi}{\partial z^*}.
\end{equation}
We can express the convergence $\kappa$ as
\begin{equation} \nonumberoff
	\kappa = \frac{\partial \alpha^*}{\partial z^*}.
\end{equation}
Furthermore, the complex shear $\gamma \equiv \gamma_1 + \mathrm{i} \gamma_2$ satisfies the relation
\begin{equation} \label{eq:complex_shear_definition}
	\gamma^* = \frac{\partial \alpha^*}{\partial z}.
\end{equation}
Using this complex formulation, we analyze the lensing properties of an elliptical Gaussian convergence next in Section \ref{sec:gaussian_lensing}. 

\subsection{Lensing by elliptical Gaussian convergence} \label{sec:gaussian_lensing}

We derive the deflection angle and shear for the elliptical Gaussian convergence
\begin{equation} \label{eq:elliptical_gauss_convergence}
	\kappa(R) = \kappa_0 \exp \left( -\frac{R^2}{2 \sigma^2}\right),
\end{equation}
where $R=\sqrt{q^2x^2 + y^2}$ is the elliptical radius. Using the complex formulation, the deflection angle for the elliptical convergence can be obtained from
\begin{equation} \label{eq:deflection_integral}
\begin{split}
	\alpha^*(z) &= 2\ \sgn (z)\ \int_0^{R(z)} \dd \zeta \frac{\zeta \kappa(\zeta)}{\sqrt{q^2 z^2 - (1-q^2) \zeta^2}} \\
	&= \frac{2 \kappa_0 \ }{ q z} \int_0^{R(z)} \dd \zeta \frac{\zeta \exp (-\zeta^2/2\sigma^2)}{\sqrt{1 - {(1-q^2)\zeta^2}/{q^2 z^2}}},
\end{split}
\end{equation}
where $\sgn(z) \equiv \sqrt{z^2}/z$ is the complex sign function, and  $R(z) = \sqrt{q^2x^2 + y^2}$ is the semi-minor axis length for the ellipse with axis-ratio $q$ that goes through the point $z=x+iy$ \citep{Bourassa75, Bray84}. With changes of variables $s=1/2\sigma^2$, $t = (1-q^2)/q^2 z^2$, $\tau = \sqrt{1 - t\zeta^2}$, we can express equation (\ref{eq:deflection_integral}) as
\begin{equation} \label{eq:deflection_conjugate}
	\begin{split}
		\alpha^*(z) &=  \frac{2\kappa_0 \ \mathrm{e}^{-s/t}}{q z t} \int_{\sqrt{1-tR(z)^2}}^{1} \dd \tau \exp \left( \frac{s}{t} \tau^2 \right) \\
		&= \frac{2\kappa_0 \ \mathrm{e}^{-s/t}}{ q z t} \left[ \frac{1}{2}  \sqrt{\frac{\uppi t}{s}} \erfi \left( \sqrt{\frac{s}{t}} \tau \right) \right]_{\sqrt{1-t(q^2x^2 + y^2)}}^{1} \\
		&=  \kappa_0 \sigma \sqrt{\frac{2 \uppi}{1-q^2}} \exp \left( -\frac{q^2 z^2}{2\sigma^2(1 - q^2)}\right) \left[ \erfi \left( \frac{q z}{\sigma \sqrt{2(1-q^2)}} \right) \right. \\
		& \qquad \qquad \qquad \qquad - \left. \erfi \left( \frac{ q^2 x + i y}{\sigma \sqrt{2(1-q^2)}} \right) \right] \\
		&= \kappa_0 \sigma \sqrt{\frac{2\uppi}{1-q^2}}\ \varsigma(z;q),
	\end{split}
\end{equation}
where $\erfi\ (z) \equiv -\mathrm{i}\ \erf\ (\mathrm{i}z)$ and we have defined the function
\begin{equation} \label{eq:ddaw_elliptical}
\begin{split}
	\varsigma (z; q) &\equiv \exp \left( -\frac{q^2 z^2}{2\sigma^2(1 - q^2)}\right) \left[ \erfi \left( \frac{q z}{\sigma \sqrt{2(1-q^2)}} \right) \right. \\
		& \qquad \qquad \qquad \qquad - \left. \erfi \left( \frac{ q^2 x + \mathrm{i} y}{\sigma \sqrt{2(1-q^2)}} \right) \right].
\end{split}
\end{equation}
We obtain the complex conjugate of the complex shear from equation (\ref{eq:complex_shear_definition}) as
\begin{equation} \label{eq:shear_conjugate}
 \begin{split}
	\gamma^*(z) &= -\frac{\kappa_0} {1-q^2} \left[ (1+q^2) \exp \left( - \frac{q^2x^2+y^2} {2\sigma^2} \right) \right.  - 2q \\ 
	&\quad\ + \frac{\sqrt{2\uppi} q^2 z}{\sigma\sqrt{1-q^2}} \exp \left(-\frac{q^2 z^2}{2\sigma^2(1-q^2)}\right) \left\{ \erfi \left( \frac{qz} {\sigma\sqrt{2(1-q^2)}}\right) \right. \\
	&\quad \left. \left. - \erfi \left( \frac{q^2x + \mathrm{i} y} {\sigma\sqrt{2(1-q^2)}}\right) \right\} \right] \\
	&= -\frac{1}{1-q^2} \left[ (1+q^2) \kappa(x, y) - 2q\kappa_0  + \frac{\sqrt{2\uppi} q^2\kappa_0 z}{\sigma \sqrt{1-q^2}}\ \varsigma(z;q) \right] .
 \end{split}
\end{equation}
Both the deflection angle and the shear contain the function $\varsigma(z;q)$. This function relates to the Faddeeva function $w_{\rm F}(z)$. First, we write the function $\varsigma(z;q)$ as
\begin{equation}
	\varsigma(z;q) = \varpi \left(\frac{q z}{\sigma \sqrt{2(1-q^2)}}; 1\right) - \varpi \left(\frac{q z}{\sigma \sqrt{2(1-q^2)}}; q \right),	
\end{equation}
where $\varpi (z;q) = \exp\left(-z^2 \right)\ \rm{erfi} \left(qx + iy/q \right)$. We can express $\varpi(z;q)$ using the Faddeeva function $w_{\rm F}(z)$ as
\begin{equation} \label{eq:daw_elliptical_to_faddeeva}
\begin{split}
	\varpi (z;q) &= 	\mathrm{e}^{-x^2 - 2\mathrm{i}xy} \mathrm{e}^{y^2} - \mathrm{i} \exp\left[ -x^2 (1-q^2) - y^2 (1/q^2 -1) \right] \\
	& \qquad \times w_{\rm F}(q x + \mathrm{i} y/q).
\end{split}
\end{equation}
Thus, we can compute the deflection angle and the shear using the Faddeeva function (Fig. \ref{fig:gaussian_lens}). Faddeeva function is a well-studied special function for its various applications in physics, for example, in radiative transfer and in plasma physics \citep[e.g.,][]{Armstrong67, Jimenez-Dominguez89}. We can readily compute $w_{\rm F}(z)$ in \textsc{python} using the \textsc{scipy.special.wofz} function. For some other popular programming languages, code-packages to compute this function are available at the web-address \url{http://ab-initio.mit.edu/Faddeeva}. In this paper, we use the algorithm outlined by \citet{Zaghloul17} to compute $w_{\rm F}(z)$ with relative error less than $4\times10^{-5}$ over the whole complex plane. We state this algorithm in Appendix \ref{app:alg_985}. A \textsc{python} implementation of this algorithm is about twice as fast as the function provided by \textsc{scipy}.  As a result, we can efficiently compute the gradient and Hessian of the deflection potential for an elliptical Gaussian convergence using equations (\ref{eq:deflection_conjugate}) and (\ref{eq:shear_conjugate}).

\begin{figure}
	\includegraphics[width=\columnwidth]{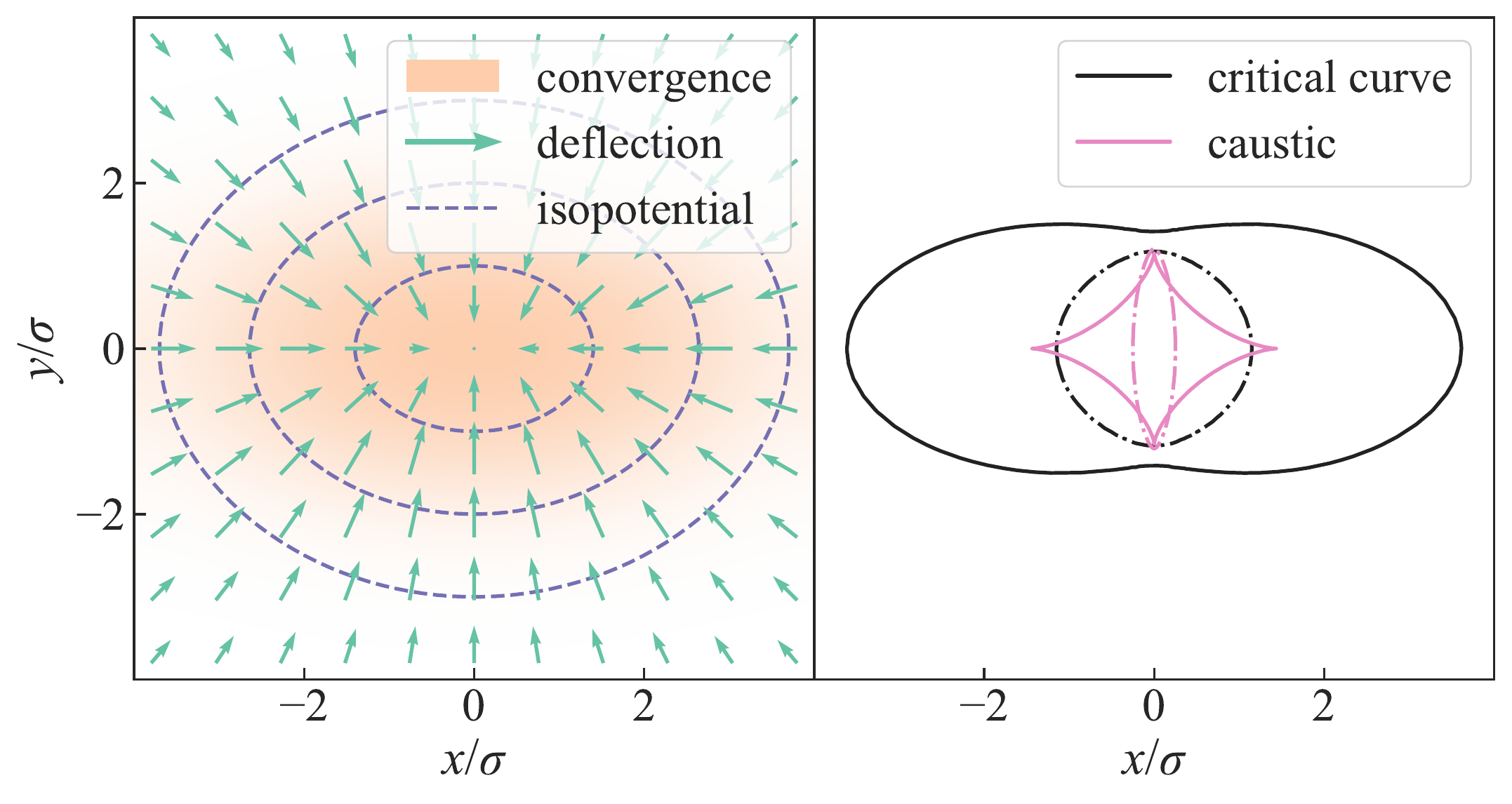}
	\caption{
	Lensing quantities for an elliptical Gaussian convergence profile. \textbf{Left:} convergence  (orange shade), deflection field (green arrows), isopotential contours (blue, dashed contours). The arrow directions are for the negative of the deflection angles and the lengths are shrunk by a factor of 4 for nicer visualization. \textbf{Right:} Critical curves (black lines) and corresponding caustics (pink lines). The solid-contour caustic corresponds to the solid-contour critical curve, similarly dot-dashed contours correspond to each other. Here, we take the amplitude of the Gaussian convergence $\kappa_0=2$ and the axis ratio $q=0.5$. We express the gradient and the Hessian of the deflection potential for an elliptical Gaussian convergence using the complex error function, as a result we can efficiently compute them.
	\label{fig:gaussian_lens}
	}
\end{figure}

Now, we turn our attention to computing the deflection potential. The deflection potential $\psi (z)$ is given by
\begin{equation}  \label{eq:complex_potential_integral} 
	\psi(z) = \re\left(\int^z_{0} \alpha^*(z^\prime)\ \dd z^\prime \right),
\end{equation}
where we set $\psi(0) = 0$. Often times we are interested in the potential difference between two points $z_1$ and $z_2$ given by
\begin{equation} \nonumberoff \label{eq:complex_potential_difference_integral}
	\Delta \psi = \psi (z_2) - \psi (z_1) = \re\left( \int_{z_1}^{z_2} \alpha^*(z^\prime)\ \dd z^\prime \right).
\end{equation}
This integral is independent of the choice of a contour. We have to carry out this integral numerically. However, the number of times we need to compute it in most applications, e.g., for computing time delays, is much fewer than that for $\alpha(z)$. We can also numerically solve the Poisson equation $\nabla^2 \psi = 2\kappa$ using the Fourier transform of the deflection potential $\hat{\psi} \equiv \mathcal{F}[\psi]$ \citep{vandeVen09}. This equation turns into $k^2 \hat{\psi} = 2\hat{\kappa}$ in the Fourier domain. The solution of the Poisson equation is then $\psi = \mathcal{F}^{-1}(2\hat{\kappa}/k^2)$. We can analytically compute the forward Fourier transform because the convergence has the Gaussian form. Then, we need to compute only the inverse transform numerically. Although obtaining the deflection potential necessitates a numerical integration or a numerical Fourier transform, we can keep the computational burden under control in most applications by computing this quantity only for a feasible number of models sampled from the lens model posterior.

\section{Adding it all together: proof of concept} \label{sec:proof_of_concept}

In this section, we demonstrate the feasibility of our method to model lenses. We first simulate synthetic data of a mock lens and then model the lens using our method. We use the publicly available lens-modelling software \textsc{lenstronomy} to simulate the synthetic data and perform the model-fitting \citep{Birrer18}. We added extra modules to \textsc{lenstronomy} to implement the lensing analysis presented in Section \ref{sec:gaussian_lensing}.

For the mock strong lensing system, we adopt an elliptical NFW deflection potential for the dark component and an elliptical Chameleon convergence for the luminous component. We take realistic scale sizes and normalizations for these profiles. We choose the Chameleon profile for two reasons:
\begin{enumerate}[topsep=0pt, leftmargin=20pt]
	\item we can analytically simulate the data with ellipticity in the convergence, and
	\item we know the S\'ersic-profile parameters that approximates the chosen Chameleon profile \textit{a priori}, so we can check the fidelity of our method.
\end{enumerate}
We parameterize the scaling of the NFW profile with two parameters: scale radius $r_{\rm s}$ and the deflection angle $\alpha_{\rm s}$ at $r_{\rm s}$. For a spherical NFW profile given by
\begin{equation} \nonumberoff
	\rho_{\rm NFW} = \frac{\rho_{\rm s}}{(r/r_{\rm s})(1 + r/r_{\rm s})^2},
\end{equation}
the normalization $\rho_s$ relates to $\alpha_s$ as
\begin{equation} \nonumberoff
	\alpha_{\rm s} = \frac{4 \rho_{\rm s} r_{\rm s}^2}{D_{\rm d}\Sigma_{\rm crit}} \left( 1 - \ln 2 \right).
\end{equation}
\citep{Meneghetti03}. The elliptical Chameleon convergence is given by
\begin{equation} \nonumberoff
\begin{split}
	\kappa_{\rm Chm} (x, y) &= \frac{\kappa_0}{1 + q} \left[ \frac{1}{\sqrt{x^2+y^2/q^2 + 4w_{\rm c}^2/(1+q^2)}} \right. \\
	 & \quad \quad \quad \quad \left. - \frac{1}{\sqrt{x^2+y^2/q^2 + 4w_{\rm t}^2/(1+q^2)}} \right]
\end{split}
\end{equation}
\citep{Suyu14}. We also add external shear to the mass profiles. Therefore, our fiducial lens mass profile has three components in total: elliptical NFW deflection potential, elliptical Chameleon convergence, and external shear. 

We simulate data for this fiducial lens system with image quality similar to the \textit{Hubble Space Telescope} (\textit{HST}) Wide-Field Camera 3 imaging in the F160W filter (see top-left panel of Fig. \ref{fig:mock_fit}). We adopt 0.08 arcsec for the pixel size, 2197 s for exposure time, and a realistic point spread function (PSF) to achieve data quality similar to the lens sample presented by \citet{Shajib19}.

We fit the synthetic data with a model composed of elliptical NFW deflection potential, elliptical S\'ersic convergence profile decomposed into Gaussians, and external shear. Note that we take ellipticity in the deflection potential for the NFW profile due to a design restriction of \textsc{lenstronomy}. However, we can also extend an elliptical NFW convergence into Gaussians for lensing analysis in principle. For now, this limitation does not affect the point of this exercise to show that lensing analysis with Gaussian components is feasible. We take the PSF as known during model-fitting for simplicity. We also separately model the lens with the fiducial mass profiles for comparison. In both cases, the parameters for the light and the luminous mass profiles are joint except for the amplitudes. We fit the model to the data by using the MCMC method. For every sample point in the parameter space, we first decompose the elliptical S\'ersic profile into Gaussian components using equations (\ref{eq:inverse_num_approx}) and (\ref{eq:gaussian_expansion_1d}). Similar to the example in Section \ref{subsubsec:forward_transform}, we take 15 Gaussian components with logarithmically spaced $\sigma$'s between $0.02R_{\rm eff}$ and $15 R_{\rm eff}$. These 15 Gaussian components approximate the S\'ersic function well within the noise level for $0.1R_{\rm eff} \leq R \leq 10R_{\rm eff}$ (Fig. \ref{fig:1d_example}). We compute the gradient and the Hessian of the deflection potential for each Gaussian component. Finally, we add the contributions from all the individual components together to obtain these quantities for the total mass profile. These total quantities are used to compute the likelihood in the MCMC method for fitting the model to the data.  
 
 Our method fits the synthetic data very well (see the `Normalized Residuals' plot in Fig. \ref{fig:mock_fit}). The fiducial S\'ersic profile parameters are also recovered with reasonable to high accuracies at the same time (Table \ref{tab:fit_params}). The total runtime is only approximately three times longer than using the fiducial model with the Chameleon profile. This loss in efficiency is a reasonable tradeoff for generality. Thus, we have demonstrated a feasible implementation of our lens-modelling method.

\begin{figure*}
	\includegraphics[width=1\textwidth]{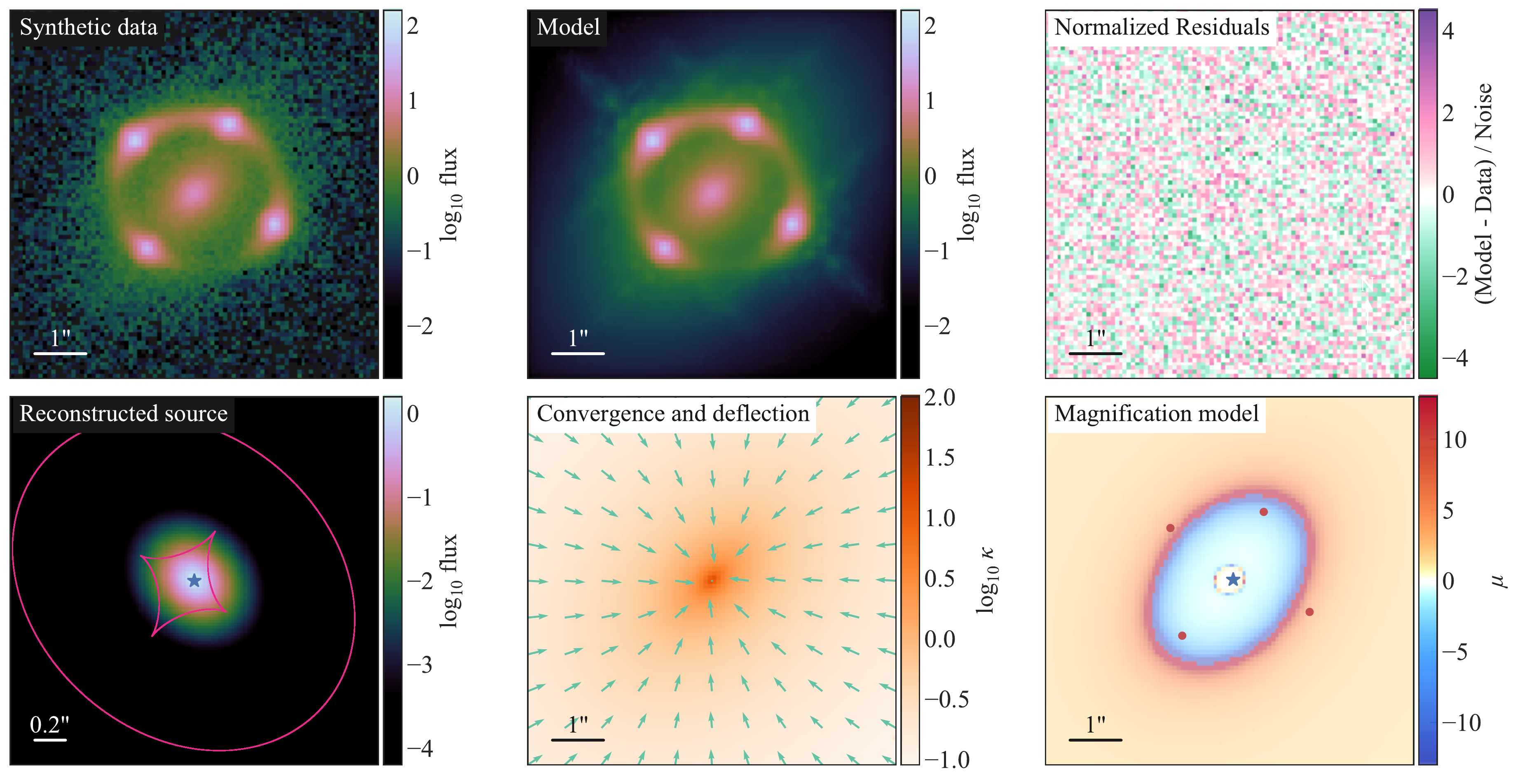}
	\caption{
	Fitting synthetic lensing data with Gaussian components of an elliptical S\'ersic profile for the luminous component. We fit the dark component with an elliptical NFW profile. The S\'ersic parameters for the lens light are joint with the luminous mass distribution except for the amplitudes letting the global mass-to-light ratio be a free parameter. We generated the synthetic data for a composite model with elliptical NFW and elliptical Chameleon profiles. We adopted a simple Gaussian point spread function. In the `Reconstructed source' plot, the pink contours outline the caustics. The blue star indicates the point source position. The green arrows in the `Convergence and deflection' plot represent negative deflection angles and they are shrunk by a factor of 4 for nicer visualization. The red dots in the `Magnification model' plot point out the image positions. Our method of computing lensing quantities for the S\'ersic profile with ellipticity in the convergence works well as evident from the `Normalized Residual' plot. This method only takes approximately three times longer than using the fiducial model with the Chameleon profile. Unlike the Chameleon profile, however, our method is general.
	\label{fig:mock_fit}
	}
\end{figure*}

\renewcommand*{\arraystretch}{1.2}
\begin{table}
	\caption{\label{tab:fit_params} 
		Fidelity of our lens modelling method. We simulate mock data with a fiducial model composed of elliptical Chameleon convergence, elliptical NFW deflection potential, and external shear. We test our method with the ``Gaussian model'': elliptical S\'ersic convergence decomposed into Gaussians, elliptical NFW deflection potential, and external shear. The `True' rows contain the mock values of the fiducial model parameters. The `Gaussian-fit' rows contain the parameters of the ``Gaussian model'' fit to the data. Similarly, the `Fiducial-fit' rows contain the parameters of the fiducial model fit to the data. We do not provide uncertainty for the values that are accurate up to the displayed decimal point. The accuracy of our computational method with Gaussian components is comparable to that using the fiducial model.
	}
	\begin{adjustbox}{width=\linewidth}
	\begin{threeparttable}
	\centering
		\begin{tabularx}{\linewidth}{lcccc}
		\hline
		Mass Profile & \multicolumn{4}{c}{Parameters} \\
		\hline
		S\'ersic & $R_{\rm eff}$ & $n_{\rm \text{S\'ersic}}$ & $q$ & $\phi$ \\
		& (arcsec) & & & (deg) \\
		\quad True\tnote{a} & 1.55 & 3.09 & 0.60 & 45  \\
		\quad Gaussian-fit & 1.51$\pm$0.02 & 3.04$\pm$0.03 & 0.61 & 45.1$\pm$0.3 \\

		Chameleon & $w_{\rm t}$ & $w_{\rm c}$ & $q$ & $\phi$ \\
		& (arcsec) & (arcsec) & & (deg)\\
		\quad True & 0.038 & 1.7 & 0.60 & 45  \\
		\quad Fiducial-fit & 0.038 & 1.7 & 0.60 & 45.0$\pm$0.2  \\
		\\
		NFW & $r_{\rm s}$ & $\alpha_{\rm s}$ & $q$ & $\phi$ \\
		& (arcsec) & (arcsec) & & (deg)\\
		\quad True & 5 & 1 & 0.9 & 45  \\
		\quad Gaussian-fit & 5.0 & 0.99$\pm$0.03 & 0.87$\pm$0.02 & 44$\pm$2 \\
		\quad Fiducial-fit & 5.0 & 1.00$\pm$0.01 & 0.90$\pm$0.01 & 45$\pm$2  \\
		\\
		External shear & $\gamma$ &$\phi$ & -- & --  \\
		& & (deg) & & \\
		\quad True & 0.051 & 5.7 & -- & -- \\
		\quad Gaussian-fit & 0.056$\pm$0.002 & 9$\pm$3 & -- & --  \\
		\quad Fiducial-fit & 0.051$\pm$0.001 & 6$\pm$2 & -- & --  \\
		\hline
	\end{tabularx}
	\begin{tablenotes}
		\item \textit{Notes.} ${}^{\text{a}}$The true values of the S\'ersic-profile parameters correspond to the true values of the Chameleon-profile parameters.
	\end{tablenotes}
	\end{threeparttable}
	\end{adjustbox}
\end{table}

\section{Conclusion: precision is feasible.} \label{sec:conclusion}

In this paper, we present a general method for precise lensing analysis of any elliptical convergence profile. Our method follows a ``divide and conquer'' strategy. In our method, we first decompose an elliptical convergence profile into Gaussian components as
\begin{equation} \label{eq:gauss_decompose_conclusion} \nonumberoff
	\kappa(x, y) \approx \sum_{j=1}^J \kappa_{0j} \exp \left( -\frac{q^2 x^2 + y^2}{2 \sigma_j^2} \right)
\end{equation}
We then compute lensing quantities, e.g., the gradient and the Hessian of the deflection potential, for each Gaussian component. Finally, we add the lensing quantities from individual Gaussian components together to obtain these quantities for the total surface-density profile. Moreover, we can straightforwardly deproject a Gaussian component to obtain its corresponding three-dimensional density profile assuming either axisymmetry or spherical symmetry. Then, we can also compute the kinematic properties, such as the line-of-sight velocity dispersion, for each Gaussian component \citep{Cappellari08}. We can then add the velocity dispersions from individual Gaussians together to obtain the total line-of-sight velocity dispersion. In this way, we self-consistently unify the lensing and kinematic descriptions of any elliptical mass profile.

We introduce an integral transform with a Gaussian kernel that leads us to a general, precise, and fast algorithm for decomposing a surface density profile into Gaussian components. Without such an algorithm, decomposing into Gaussians would end up as a bottleneck in the lens modelling efficiency. We obtain the algorithm by first inverting the integral transform as

\begin{equation} \nonumberoff
	f(\sigma) = \frac{1}{\mathrm{i}\sigma^2} \sqrt{\frac{2}{\uppi}} \int_{C} z F(z) \exp \left( \frac{z^2}{2 \sigma^2} \right) \dd z.
\end{equation}
Although this is an integral, we provide a straightforward formula to compute $f(\sigma)$. The computed values of $f(\sigma)$ then quantify the amplitudes $\kappa_{0j}$ of the Gaussian components in equation (\ref{eq:gauss_decompose_conclusion}). As a result, this integral transform fulfills the three requirements for a decomposition algorithm to be \textbf{(i)} general, \textbf{(ii)} precise, and \textbf{(iii)} fast. To be specific, this decomposition algorithm is $\sim 10^3$ times faster than the MGE algorithm from \citet{Cappellari02}. Consequently, our lensing analysis requires the same order of CPU time as other methods currently in use to model a lens with a composite mass profile. Thus, the integral transform enables the lens modelling with the Gaussian components to be efficient and, in turn, makes our unified framework for lensing and kinematic analysis of an elliptical mass profile feasible.

Our method enables precise lens modelling with an elliptical mass profile for several astrophysical applications. Specifically, our method gives an efficient method to model composite mass profiles with separate components for the baryon and the dark matter. For example, the usual choices for these components are the S\'ersic and the NFW profiles; both are computationally difficult to directly implement in lens modelling for the elliptical case. Our method makes both of these profiles computationally tractable while achieving the required precision. Thus, our method will be useful in applications where a composite mass profile is essential for lens modelling, for example, in  detecting dark-matter substructure, in measuring the Hubble constant, and in testing massive elliptical-galaxy formation theories \citep[e.g.,][]{Vegetti12, Wong17, Nightingale19}.

\section*{Acknowledgements}
AJS thanks Simon Birrer and Shouman Das for helpful discussions. AJS additionally thanks the anonymous referee for very useful suggestions that improved this paper. AJS also expresses gratitude to Adriano Agnello, Simon Birrer, Xuheng Ding, Xinnan Du, Abhimat K. Gautam, Briley Lewis, Michael Topping, and Tommaso Treu for providing feedbacks that greatly improved the writing and the presentation of this paper. AJS acknowledges support by National Aeronautics and Space Administration (NASA) through Space Telescope Science Institute grant HST-GO-15320.

This research made use of \textsc{lenstronomy} \citep{Birrer18}, \textsc{numpy} \citep{Oliphant15}, \textsc{scipy} \citep{Jones01}, \textsc{jupyter} \citep{Kluyver16}, and \textsc{matplotlib} \citep{Hunter07}.



\appendix

\section{Properties of the Integral transform with a Gaussian Kernel} \label{app:theorems}

In this appendix, we prove some fundamental properties of the integral transform with a Gaussian kernel. In three theorems, we prove that 
\begin{enumerate}[topsep=0pt, leftmargin=25pt]
	\item the integral transform exists for a function with certain characteristics,
	\item the transform is unique for a continuous function, and
	\item the transform is invertible.
\end{enumerate}
First, we define the integral transform.

\newtheorem{thm}{Theorem}[section]
\newtheorem{defn}[thm]{Definition}
\begin{defn} \label{def:transform}
Define an integral transform $\mathcal{T}$ that takes a function $f(\sigma):\mathbb{R}_{\geq 0} \to \mathbb{R}$ to a function $F(z):\mathbb{C} \to \mathbb{C}$ as
	\begin{equation} \label{eq:transform}
		F(z) \equiv \mathcal{T}[f] (z) \equiv \frac{1}{\sqrt{2 \uppi}} \int_0^\infty \frac{f(\sigma)}{\sigma} \exp\left( -\frac{z^2}{2\sigma^2}\right) \dd \sigma.
	\end{equation}
\end{defn}
Next, we define a \textit{transformable} function in the context of this paper.
\begin{defn} \label{def:transformable_function}
A function $f(\sigma):\mathbb{R}_{\geq 0} \to \mathbb{R}$ is said to be {transformable}, if it satisfies the following conditions:
	\begin{enumerate}[topsep=0pt, leftmargin=25pt]
		\item the function $f(\sigma)$ is piecewise continuous,
		\item the function $f(\sigma) = \mathcal{O}(\exp(c/2\sigma^2))$ as $\sigma \to 0$, where $c \in  \mathbb{R}_{\geq 0}$,
		\item the function $f(\sigma) = \mathcal{O}(\sigma^\lambda)$ with $\lambda < 0 $ as $\sigma \to \infty$.
	\end{enumerate}
We refer to these three conditions as the \textit{transformability conditions}. The namesake for the transformable function is made clear next in Theorem \ref{thm:transform}.
\end{defn}

\begin{thm}[Existence] \label{thm:transform}
If $f(\sigma)$ is transformable, then its integral transform $F(z)$ exists in the region of convergence (ROC) $\re(z^2) > c$.
\end{thm}
\begin{proof}
Divide the integral in equation (\ref{eq:transform}) as
\begin{equation} \nonumberoff
\begin{split}
F(z) &= \frac{1}{\sqrt{2 \uppi}} \left( \int_0^a \dd \sigma + \int_a^b \dd \sigma + \int_b^\infty \dd \sigma \right) \frac{f(\sigma)}{\sigma} \exp\left( -\frac{z^2}{2\sigma^2}\right) \\
&= \frac{1}{\sqrt{2 \uppi}} \left( \mathcal{I}_1 + \mathcal{I}_2 + \mathcal{I}_3 \right),
\end{split}
\end{equation}
where $0<a<b<\infty$.
\begin{enumerate}[leftmargin=15pt]
	\item The integral $\mathcal{I}_2$ converges, as the integrand is piecewise continuous in $[a, b]$ according to transformability condition (i) in Definition \ref{def:transformable_function}.
	\item According to transformability condition (ii), there exists $M_1 \in \mathbb{R}_{>0}$ such that $f(\sigma) \leq M_1 \exp (c/2\sigma^2)$ for $\sigma \leq a$. Then using Jensen's inequality, we have
	\begin{equation} \label{eq:convergence_at_zero}
	\begin{split}
	\abs{\mathcal{I}_1} &\leq \int_0^a \abs{\frac{f(\sigma)}{\sigma} \exp\left( -\frac{z^2}{2\sigma^2}\right)} \dd \sigma \\
	\Rightarrow \abs{\mathcal{I}_1} &\leq \int_0^a \abs{\frac{M_1}{\sigma} \exp\left( - \frac{z^2 - c}{2\sigma^2}\right)} \dd \sigma  \\
	  \Rightarrow \abs{\mathcal{I}_1} &\leq M_1 \int_0^a \frac{1}{\sigma} \exp \left(-\frac{\re (z^2) - c}{2 \sigma^2} \right) \dd \sigma.
	\end{split}
	\end{equation}
Therefore, the integral $\mathcal{I}_1$ converges in the region $\re (z^2) = x^2 - y^2 > c$, where $z = x + \mathrm{i}y$, $x \in \mathbb{R}, y \in \mathbb{R}$.
	\item According to transformability condition (iii), there exists $M_2 \in \mathbb{R}_{>0}$ such that $f(\sigma) \leq M_2 \sigma^{\lambda}$ for $\sigma \geq b$. Then, we have
	\begin{equation} \label{eq:convergence_at_infinity} \nonumberoff
	\begin{split}
		\abs{\mathcal{I}_3} &\leq \int_b^\infty \abs{\frac{f(\sigma)}{\sigma} \exp \left( - \frac{z^2}{2 \sigma^2} \right) } \dd \sigma \\
		\Rightarrow \abs{\mathcal{I}_3} &\leq \int_b^\infty M_2 \sigma^{\lambda-1} \abs{\exp \left( - \frac{z^2}{2 \sigma^2}\right)} \dd \sigma \\
		\Rightarrow \abs{\mathcal{I}_3} &\leq M_2 \int_b^\infty \sigma^{\lambda-1} \dd \sigma = -M_2\frac{b^{\lambda}}{\lambda} < \infty.
	\end{split}
	\end{equation}
Here, we applied the inequality $\abs{\exp (-z^2/2\sigma^2)} \leq 1$ for $\re (z^2) > c \geq 0$. As a result, the integral $\mathcal{I}_3$ converges.
\end{enumerate}
Therefore, the transform $F(z)$ exists in the ROC $\re(z^2) > c$.
\end{proof}

Fig. \ref{fig:roc} shows the ROC for the integral in equation (\ref{eq:transform}). We can extend the ROC by the following two corollaries.

\begin{figure}
\includegraphics[width=0.5\textwidth]{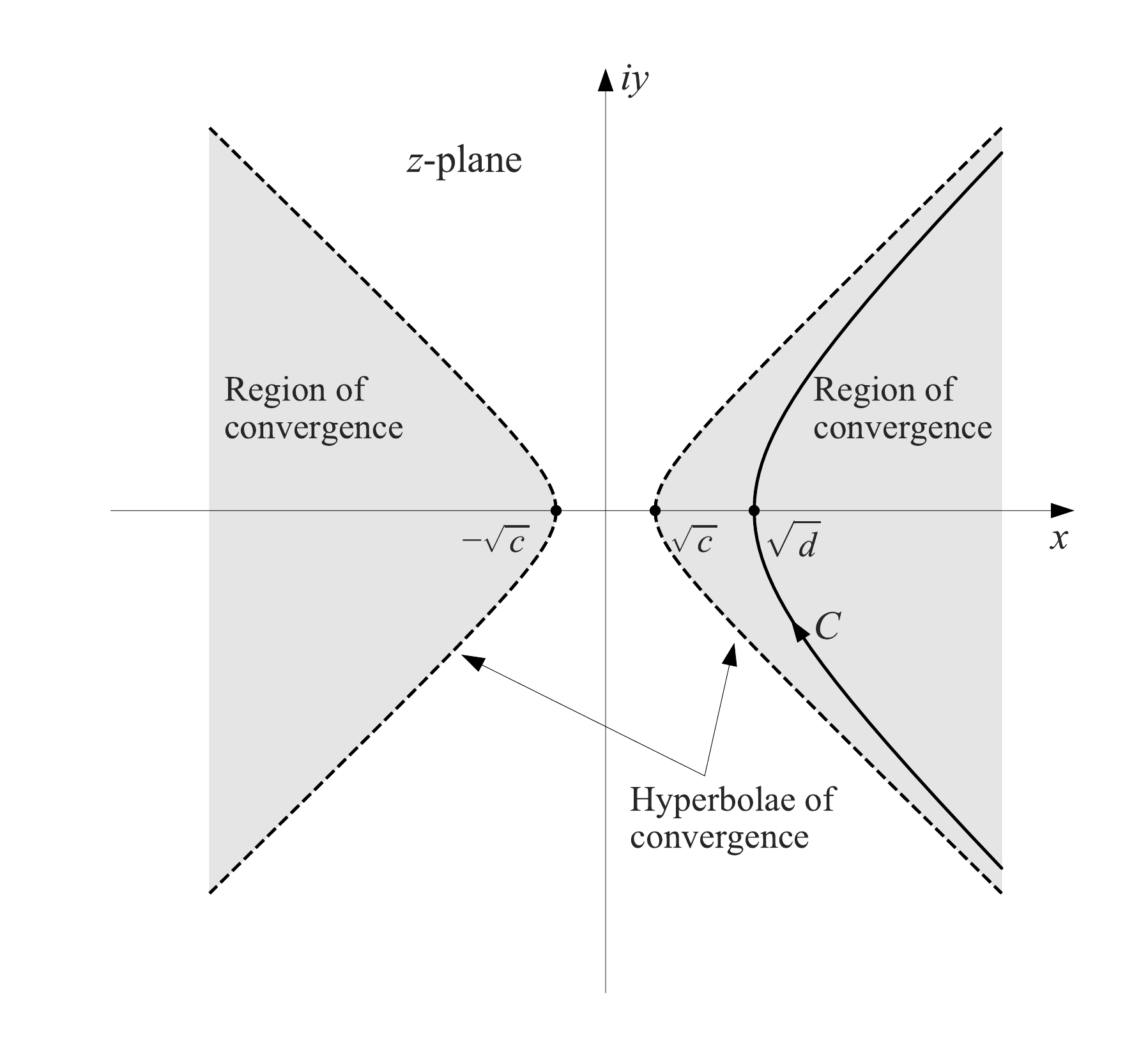}
\caption{Region of convergence (shaded region) on the complex plane for the integral in equation (\ref{eq:transform}). The hyperbolic contour $C$ for the integral in equation (\ref{eq:inverse_transform}) is shown with solid black hyperbola.
	\label{fig:roc}
}
\end{figure}

\newtheorem{crl}[thm]{Corollary}
\begin{crl} \label{crl:equality_transform}
If a transformable function $f(\sigma)$ additionally satisfies the condition $f(\sigma) = \mathcal{O} (\sigma^\beta \exp(c/2\sigma^2))$ with $\beta \geq 1$ as $\sigma \to 0$, then the integral in equation (\ref{eq:transform}) converges in the ROC $\re (z^2) \geq c$.
\end{crl}

\begin{proof}
According to the additional condition, there exists $ M_3 \in \mathbb{R}_{>0}$ such that $f(\sigma) \leq M_3 \sigma^{\beta} \exp(c/2\sigma^2)$ for $\sigma \leq a$. Then, we can rewrite equation (\ref{eq:convergence_at_zero}) as
\begin{equation} \label{eq:corollary_1} \nonumberoff
	\begin{split}
	\abs{\mathcal{I}_1} &\leq \int_0^a \abs{\frac{f(\sigma)}{\sigma} \exp\left( -\frac{z^2}{2\sigma^2}\right)} \dd \sigma \\
	\Rightarrow \abs{\mathcal{I}_1} &\leq \int_0^a \abs{M_3 \sigma^{\beta - 1} \exp\left( - \frac{z^2 - c}{2\sigma^2}\right)} \dd \sigma  \\
	  \Rightarrow \abs{\mathcal{I}_1} &\leq M_3 \int_0^a \sigma^{\beta-1} \exp \left(-\frac{\re (z^2) - c}{2 \sigma^2} \right) \dd \sigma.
	\end{split}
\end{equation}
For $\re (z^2) = c$, this becomes
\begin{equation} \nonumberoff
\begin{split}
\abs{\mathcal{I}_1} &\leq M_3 \int_0^a \sigma^{\beta - 1} \dd \sigma = \frac{M_3 a^{\beta}}{\beta} < \infty.
\end{split}
\end{equation}
Therefore, the ROC for the integral in equation (\ref{eq:transform}) extends to $\re (z^2) \geq c$.
\end{proof}

\begin{crl} \label{crl:at_zero_transform}
If a transformable function $f(\sigma)$ additionally satisfies $f(\sigma) = \mathcal{O} (\sigma^\beta)$ with $\beta \geq 1$ as $\sigma \to 0$, then the integral in equation (\ref{eq:transform}) converges in the ROC $\re (z^2) \geq 0$.
\end{crl}

\begin{proof} The proof is trivial with the substitution $c = 0$ in Corollary \ref{crl:equality_transform}.
\end{proof}

Next we prove the uniqueness theorem for the transform. First, we state a well-known proof of the following lemma for completeness.

\newtheorem{lma}[thm]{Lemma}
	\begin{lma} \label{lma:zero_function}
		If $f(x)$ is continuous in $[0, 1]$, and $\int_0^1 x^n f(x)\ \dd x = 0$ for $n=0,1,2,\dots$, then $f(x)=0$.
	\end{lma}
	\begin{proof}
		From the Weierstrass approximation theorem, for any $ \epsilon>0$, there exists a polynomial $P_{\epsilon}(x)$ such that $\abs{f(x)-P_{\epsilon}(x)}< \epsilon$ for all $x \in [0,1]$. The hypothesis implies that $\int_0^1 P_{\epsilon}(x) f(x)\ \dd x  = 0$. By taking the limit $\epsilon \to 0$, this equation becomes $\int_0^1 f(x)f(x)\  \dd x = 0$. As $f(x)^2 \geq 0$, we have $f(x)=0$. 
	\end{proof}
	
\begin{thm}[Uniqueness] \label{thm:uniqueness}
If $f(\sigma)$ and $g(\sigma)$ are continuous, and $\mathcal{T}[f](z) = \mathcal{T}[g](z)$ for all $z$ in the ROC, then $f(\sigma)=g(\sigma)$.
\end{thm}
\begin{proof}
	Due to linearity, it is sufficient to prove that if $\mathcal{T}[f](z) = 0$, then $f(\sigma)=0$. Take $d$ such that the contour $\re (z^2) = d$ lies in the ROC. By making the change of variables $s=\exp(-1/2\sigma^2)$, for $z^2=d+n+1$ with $n=0,1,2,\dots$, we have
	\begin{equation} \nonumberoff
		\begin{split}
			&\mathcal{T}[f](z) = \frac{1}{\sqrt{2\uppi}} \int_0^\infty \frac{f(\sigma)}{\sigma} \exp \left( -\frac{d+n+1}{2\sigma^2}\right) \dd \sigma = 0 \\
			\Rightarrow &\int_0^1 \left[ - \frac{s^d f(\sqrt{-1/2\log s})}{2 \sqrt{2\uppi} \log s} \right] s^n \dd s = 0.
		\end{split}
	\end{equation}
	This integral exists as $s \to 0$, because
	\begin{equation} \nonumberoff
		\lim_{s \to 0}	\left[ - \frac{s^d f(\sqrt{-1/2\log s})}{2 \sqrt{2\uppi} \log s} \right] = \lim_{\sigma \to 0} \left[ \frac{\sigma^2 f(\sigma)}{\sqrt{2 \uppi}} \exp \left(- \frac{d}{2\sigma^2} \right) \right] = 0.
	\end{equation}
	Therefore, according to Lemma \ref{lma:zero_function}, we have $f(\sigma)=0$.
\end{proof}

\begin{thm}[Inversion] \label{thm:inverse_transform}
If $F(z)$ is the transform of $f(\sigma)$, then $f(\sigma)$ is given by the inverse transform
\begin{equation} \label{eq:inverse_transform}
	f(\sigma)= \mathcal{T}^{-1}[F](\sigma) = \frac{1}{\mathrm{i}\sigma^2} \sqrt{\frac{2}{\uppi}} \int_{C} z F(z) \exp \left( \frac{z^2}{2 \sigma^2} \right) \dd z,
\end{equation}
where the contour $C$ is the hyperbola $\re (z^2) = d$ such that $C$ lies in ROC of $F(z)$.
\end{thm}

\begin{proof}
Write equation (\ref{eq:transform}) for $z^2 = d$ as
\begin{equation}\nonumberoff
	F\left(\sqrt{d}\right) = \frac{1}{\sqrt{2 \uppi}} \int_0^\infty \frac{f(\sigma)}{\sigma} \exp\left( -\frac{d}{2\sigma^2}\right) \dd \sigma .
\end{equation}
With the change of variables $p = 1/\sigma^2$, this equation transforms into
\begin{equation} \label{eq:laplace_transform}
	F\left(\sqrt{d} \right) = \int_0^\infty g(p) \exp \left( -\frac{dp}{2} \right) \dd p,
\end{equation}
where
\begin{equation} \nonumberoff
g(p) = \frac{\sigma^2 f(\sigma)}{2 \sqrt{2 \uppi}}.
\end{equation}
Define a new function
\begin{equation} \nonumberoff
h(p) = 
     \begin{cases}
       g(p) \exp \left( -\frac{dp}{2}\right), & p \geq 0,\\
       0, &p < 0. \\ 
     \end{cases}
\end{equation}
Theorem \ref{thm:transform} implies that $\int_{-\infty}^\infty \abs{h(p)} \ \dd p < \infty $, thus $h(p)$ belongs to the Lebesgue space $L^1 (\mathbb{R})$. Therefore, we can take the Fourier transform of $h(p)$ as
\begin{equation} \nonumberoff
\begin{split}	
	\hat{h}(\nu) &= \frac{1}{\sqrt{2 \uppi}} \int_{-\infty}^{\infty} h(p) \mathrm{e}^{-\mathrm{i}p\nu} \dd p \\
	&=\frac{1}{\sqrt{2 \uppi}} \int_{0}^{\infty} g(p) \exp \left( -\frac{(d+2\mathrm{i}\nu)p}{2}\right) \dd p \\
	&= \frac{1}{\sqrt{2 \uppi}} F\left( \sqrt{d+2\mathrm{i}\nu} \right),
\end{split}
\end{equation}
where we used equation (\ref{eq:laplace_transform}) for the substitution in the last line. Now, take the inverse Fourier transform of $\hat{h}(\nu)$ as 
\begin{equation} \nonumberoff
\begin{split}
	&h(p) = \frac{1}{\sqrt{2\uppi}} \int_{-\infty}^{\infty} \hat{h} (\nu) e^{\mathrm{i} p \nu} \dd \nu \\
	\Rightarrow &g(p) \exp \left( - \frac{dp}{2} \right) = \frac{1}{2 \uppi} \int_{-\infty}^{\infty} F\left( \sqrt{d + 2 \mathrm{i} \nu} \right) e^{\mathrm{i} p \nu} \dd \nu \\
	\Rightarrow &\frac{\sigma^2 f(\sigma)}{2 \sqrt{2 \uppi}} = \frac{1}{2 \uppi} \int_{-\infty}^{\infty} F\left( \sqrt{d + 2 \mathrm{i} \nu} \right) \exp \left(\frac{d + 2 \mathrm{i} \nu}{2\sigma^2} \right) \dd \nu \\
	\Rightarrow &f(\sigma) = \frac{1}{\mathrm{i}\sigma^2} \sqrt{\frac{2}{\uppi}} \int_{C} z F(z) \exp \left( \frac{z^2}{2 \sigma^2} \right) \dd z.
\end{split}
\end{equation}
Here, we used the substitution of variable $z^2 = d+2 \mathrm{i} \nu$ in the last line, which transforms the integral path to the hyperbolic contour $C$ given by $\re (z^2) = d$.
\end{proof}

\newtheorem{rmk}[thm]{Remark}
\begin{rmk} \label{rmk:laplace}
Equation (\ref{eq:laplace_transform}) has the form of a Laplace transform. Therefore, the integral transform with a Gaussian kernel for a transformable function can be converted into a Laplace transform by suitable change of variables.
\end{rmk}

\section{Efficient algorithm to compute the Faddeeva function} \label{app:alg_985}
In this appendix, we state an efficient algorithm to compute the Faddeeva function $w_{\rm F}(z)$ \citep[Algorithm \ref{alg:alg_985}, for details see][]{Zaghloul17}. The relative error of this algorithm is less than  $4\times 10^{-5}$ over the whole complex plane.

\algnewcommand{\newAnd}{\text{ and }}
\algnewcommand{\Or}{\text{ or }}
\begin{algorithm}
	\caption{Compute $w_{\rm F}(z)$}
	\label{alg:alg_985}
	\begin{algorithmic}
		\State $y \gets \im (z)$ 
		\If {$\abs{z}^2 \geq 3.8\times 10^{4}$}
	    	\State $w_{\rm F} \gets \dfrac{\rm i}{ z \sqrt{\uppi}}$
		\ElsIf {$3.8\times 10^{4} > \abs{z}^2 \geq 256$}
    	    \State $w_{\rm F} \gets \dfrac{{\rm i} z} {\sqrt{\uppi} (z^2 - 0.5)}$ 
    	\ElsIf {$256 > \abs{z}^2 \geq 62$}
    	    \State $w_{\rm F} \gets \dfrac{{\rm i} (z^2 - 1) }{ z\sqrt{\uppi} (z^2 - 1.5)}$
    	\ElsIf {$62 > \abs{z}^2 \geq 30 \And y^2 \geq 10^{-13}$}
    	    \State $w_{\rm F} \gets \dfrac{{\rm i} z (z^2 - 2.5) }{ \sqrt{\uppi} (z^2(z^2 - 3)+0.75)}$
    	\ElsIf {$(62 > \abs{z}^2 \geq 30 \newAnd y^2 < 10^{-13})\ \Or\ (30>\abs{z}^2\geq2.5 \newAnd y^2<0.072)$}
    	    \State $ {U} \gets [1.320522,\ 35.7668,\ 219.031, 1540.787$\newline
    	    		\hspace*{3.85em} $ 3321.990,\ 36183.31]$
    	    \State $ {V} \gets [1.841439,\ 61.57037,\ 364.2191,\ 2186.181$\newline
    	    		\hspace*{3.85em} $9022.228,\ 24322.84,\ 32066.6]$
    	    \State $w_{\rm F} \gets \exp(-z^2) + {\rm i}z \times (U[6]+z^2(U[5]+z^2(U[4]+z^2(U[3]+$\newline\hspace*{4.3em}$z^2(U[2]+z^2(U[1]+z^2\sqrt{\uppi}))))))/(V[7]+z^2(V[6]+$\newline\hspace*{4.3em}$z^2(V[5]+z^2(V[4]+z^2(V[3]+z^2(V[2]+z^2(V[1]$\newline\hspace*{4.3em}$+z^2)))))))$
    	\Else
    		\State $ {U} \gets [5.9126262,\ 30.180142,\ 93.15558,\ 181.92853$\newline
    	    		\hspace*{3.85em} $214.38239,\ 122.60793]$
    	    \State $ {V} \gets [10.479857,\ 53.992907,\ 170.35400,\ 348.70392$\newline
    	    		\hspace*{3.85em} $457.33448,\ 352.73063,\ 122.60793]$
    	    \State $w_{\rm F} \gets (U[6]-{\rm i}z(U[5]-{\rm i}z(U[4]-{\rm i}z(U[3]-{\rm i}z(U[2]- {\rm i}z(U[1]$ \newline\hspace*{4.3em}$- {\rm i}z\sqrt{\uppi}))))))/ \
                (V[7]-{\rm i}z(V[6]-{\rm i}z(V[5]-{\rm i}z(V[4]$\newline\hspace*{4.3em}$-{\rm i}z(V[3]-{\rm i}z(V[2]-{\rm i}z(V[1]-{\rm i}z)))))))$
		\EndIf
	\end{algorithmic}
\end{algorithm}






\bibliographystyle{mnras}
\bibliography{/Users/ajshajib/mylib/STRIDES/Papers/Anowar/ajshajib}

\bsp	
\label{lastpage}

\end{document}